\documentclass{scrartcl}
\usepackage{amsmath}
\usepackage{amsthm}
\usepackage{amssymb}
\usepackage{bbold}

\usepackage{bookmark}
\usepackage{enumitem}

\newtheorem{lemma}{Lemma}[section]
\newtheorem{theorem}[lemma]{Theorem}
\newtheorem{corollary}[lemma]{Corollary}
\newtheorem{conjecture}[lemma]{Conjecture}
\newtheorem{question}[lemma]{Question}
\newtheorem{proposition}[lemma]{Proposition}

\theoremstyle{definition}
\newtheorem{definition}[lemma]{Definition}

\newcommand\amo{\operatorname{\mathrm{AMO}}}
\newcommand\exone{\operatorname{\mathrm{EO}}}
\newcommand\penc{P-encoding}
\newcommand\pencs{P-encodings}
\newcommand\pencS[1]{\mathcal{P}(#1)}
\newcommand\pencQS[1]{\mathcal{P}_2(#1)}
\newcommand\amoS[1]{\mathcal{A}(#1)}
\newcommand\exoneS[1]{\mathcal{E}(#1)}

\newcommand\var[1]{\mathrm{var}(#1)}
\newcommand\lit[1]{\mathrm{lit}(#1)}
\newcommand\resolvent[2]{{\cal R}(#1, #2)}
\newcommand\upclos[2][]{{\cal U}_{#1}(#2)}
\newcommand\vek[1]{\mathbf{#1}}

\newcommand{\PA[2]}[\varphi]{L_{#1,#2}}
\newcommand{\PB[2]}[\varphi]{L^\mathrm{v}_{#1,#2}}
\newcommand{\PC[2]}[\varphi]{K_{#1,#2}}
\newcommand{\QA[2]}[\varphi]{Q_{#1,#2}}
\newcommand\partass{\rho}
\renewcommand\phi{\varphi}
\newcommand\Rtype{R} % Instead of \xspace use \Rtype{}
\newcommand\Qtype{Q} % Instead of \xspace use \Qtype{}

\newcommand\subst[3]{#1[#2\leftarrow #3]}

\newenvironment{amocond}{%
   \begin{enumerate}[label=(P\arabic*)]
      }{%
   \end{enumerate}
}
\newenvironment{regcond}{%
   \begin{enumerate}[label=(R\arabic*)]
   }{%
\end{enumerate}
}

\title{A Lower Bound on CNF Encodings of the At-Most-One Constraint}
\author{Petr Ku{\v c}era\thanks{Department of Theoretical Computer Science
and Mathematical Logic,
Faculty of Mathematics and Physics,
Charles University, Czech Republic,
kucerap@ktiml.mff.cuni.cz}
   \and Petr Savick{\'y}\thanks{Institute of
      Computer Science,
   The Czech Academy of Sciences,
   Czech Republic,
   savicky@cs.cas.cz}
\and Vojt{\v e}ch Vorel\thanks{Department of Theoretical Computer Science
and Mathematical Logic,
Faculty of Mathematics and Physics,
Charles University, Czech Republic,
vorel@ktiml.mff.cuni.cz}
   }
\date{}

\begin{document}

   \maketitle

   \begin{abstract}
      Constraint ``at most one'' is a basic cardinality constraint
      which requires that at most one of its \(n\) boolean inputs is
      set to \(1\). This constraint is widely used when translating
      a problem into a \emph{conjunctive normal form} (CNF) and we
      investigate its CNF encodings suitable for this purpose. An
      encoding differs from a CNF representation of a function in that
      it can use auxiliary variables. We are especially interested in
      propagation complete encodings which have the property that unit
      propagation is strong enough to enforce consistency on input
      variables. We show a lower bound on the number of clauses in any
      propagation complete encoding of the ``at most one'' constraint. The
      lower bound almost matches the size of the best known encodings.
      We also study an important case of 2-CNF encodings where we show
      a slightly better lower bound. The lower bound holds also for
      a related ``exactly one'' constraint.
   \end{abstract}

   \section{Introduction}

In this paper we study the properties of one of the most basic
cardinality constraints --- the ``at most one'' constraint on \(n\)
boolean variables which requires that at most one input variable is set
to \(1\). This constraint is widely used when translating a problem
into a \emph{conjunctive normal form} (CNF). Since the ``at most one''
constraint is anti-monotone, it has a unique minimal prime CNF representation
which requires \(\binom{n}{2}=\Theta(n^2)\) negative clauses, where \(n\) is the
number of input variables. However, there are CNF encodings of size $O(n)$
which use additional auxiliary variables. Several encodings for this
constraint were considered in the literature. Let us mention \emph{sequential
   encoding}~\cite{ccSIN_2005} which addresses also more general
cardinality constraints. The same encoding was also called \emph{ladder
   encoding} in~\cite{amoHOL_2013}, and it forms the smallest variant
of the \emph{commander-variable encodings}~\cite{amoKLI_2007}.
After a minor simplification, it
requires \(3n-6\) clauses and \(n-3\) auxiliary variables. Similar,
but not smaller encodings can also be obtained as special cases of
\emph{totalizers}~\cite{ccBAI_2003} and \emph{cardinality
   networks}~\cite{ccASI_2011}.
Currently the smallest known encoding is the \emph{product encoding}
introduced by Chen~\cite{amoCHE_2010}. It consists of \(2n + 4
   \sqrt{n} + O(\sqrt[4]{n})\) clauses and uses \(O(\sqrt{n})\)
auxiliary variables.
The sequential and the product encodings are described in
Section~\ref{ssec:sequential-enc} and
Section~\ref{sec:product-enc} with some modifications.
It is worth noting that the product encoding can be derived using
the monotone circuit of size $kn + o(n)$ for the function $T^n_k$
described in~\cite{D84} and in~\cite{W87}, if $k=2$.
Section~\ref{sec:monotone-circuits} provides more detail on this.

Other encodings introduced in the literature for the ``at most one''
constraint use more clauses than
either sequential or product encoding does. These include the \emph{binary
encoding}~\cite{ccHAI_2012,amoFRI_2005} and the \emph{bimander
encoding}~\cite{amoHOL_2013}.

All the encodings for the ``at most one'' constraint we have mentioned
are in the form of a 2-CNF formula which is
a CNF formula where all clauses consist of at most two literals. This restricted
structure guarantees that the encodings are propagation complete.
The notion of propagation
completeness was introduced by~\cite{BM12} as a generalization of unit
refutation completeness introduced by~\cite{V94}. We say that a
formula \(\phi\) is \emph{propagation complete} if for any set of
literals \(e_i,\; i \in I\), the following property holds: either
\(\phi\wedge\bigwedge_{i \in I} e_i\) is contradictory and this can
be detected by unit propagation, or unit propagation started with
\(\phi\wedge\bigwedge_{i \in I} e_i\) derives all literals \(g\)
that logically follow from this formula. It was shown
in~\cite{BBCGK13} that a prime 2-CNF formula is always propagation complete.
Since unit propagation is a standard tool used in state-of-the-art
SAT solvers~\cite{BHMW09}, this makes 2-CNF formulas as a part of 
a larger instance simple for them.

When encoding a constraint into a CNF formula, a weaker condition
than propagation completeness of the resulting formula is often
required. Namely, we require that unit propagation on the encoding
is strong enough to
enforce some kind of local consistency, for instance generalized arc
consistency (GAC), see for example~\cite{B07}. In this case we only
care about propagation completeness with respect to input variables
and not necessarily about behaviour on auxiliary variables. Later we
formalize this notion as \emph{propagation complete encoding} (PC encoding).
Let us note that this name was also used in~\cite{BHKM16} to denote an
encoding of a given constraint which is propagation complete with
respect to all variables including the auxiliary ones.

Chen~\cite{amoCHE_2010} conjectures that the product encoding is the
smallest possible PC encoding of the ``at most one'' constraint.
In this paper we provide support for the positive
answer to this conjecture. Our lower bound almost matches the size of
the product encoding. We show that any propagation complete
encoding of the ``at most one'' constraint on \(n\) variables requires
at least \(2n + \sqrt{n} - O(1)\) clauses. The lower bound actually
holds for a related constraint ``exactly one'' as well. We also consider the
important special case of 2-CNF encodings for which we achieve a
better lower bound, namely, any 2-CNF encoding of the ``at most one'' constraint
on \(n\) variables requires at least \(2n+2\sqrt{n}-O(1)\) clauses.

We should note that having a smaller encoding is not necessarily an advantage
when a SAT solver is about to be used. Adding auxiliary variables can
be costly because the SAT solver has to deal with them and
possibly use them for decisions. However, encodings using auxiliary variables
can be useful for constraints whose CNF representation is too large.
Moreover, the experimental results in~\cite{ML07} suggest that a SAT solver
can be modified to minimize the disadvantage of introducing auxiliary
variables. Another experimental evaluation of various
cardinality constraints and their encodings appears in~\cite{FG10}.
A propagation complete
encoding can also be used as a part of a general purpose CSP solver
where unit propagation can serve as a propagator of GAC,
see~\cite{B07}.

The paper is organized as follows. In Section~\ref{sec:results}, we
give the necessary definitions and formulate the main result
in Theorem~\ref{thm:lower-bound}. In particular, we introduce
the notion of a \penc{} that captures the common properties
of propagation complete encodings of the ``at most one''
and the ``exactly one'' constraints used for the lower bounds.
Moreover, we define a specific form of a \penc{} which we call
a regular form and formulate Theorem~\ref{thm:structure-general-CNF-2}
that is the basis of the proofs of
the lower bounds by considering separately the encodings in
the regular form and the encodings not in this form.
We also prove that Theorem~\ref{thm:structure-general-CNF-2} is sufficient for the lower bound $2n$ on
the size of the considered encodings.
In Section~\ref{sec:def}, we recall the known results and present some
auxiliary results we use in the rest of the paper.
In Section~\ref{sec:regular-form}, we prove the properties
of the encodings not in the regular form that imply
Theorem~\ref{thm:structure-general-CNF-2}.
Section~\ref{sec:gen-lb} contains the proof of a lower bound
$2n + \sqrt{n} - O(1)$ on the
size of any propagation complete encoding of the ``at most one''
and the ``exactly one'' constraints obtained by
analysis of the encodings in the regular form.
In Section~\ref{sec:quad-lb} we prove a lower bound
$2n + 2 \sqrt{n} - O(1)$
on the size of 2-CNF encodings of the ``at most one'' constraint
by a different analysis of the encodings in the regular form.
We close the paper with notes on
possible directions for
further research in Section~\ref{sec:further} and concluding
remarks in Section~\ref{sec:conclusion}.

A preliminary version of this paper appeared in~\cite{KSV17}. Due to
page limitations, several proofs were omitted or only sketched in the
conference version. In this version of the paper, we have included all
proofs and improved their readability. The lower bounds
 were slightly improved since the conference version as well.

   \section{Definitions and Results}
\label{sec:results} % chktex 24

In this section we introduce the notions used throughout the paper,
state the main results and give an overview of their proof.
We use $\subset$ to denote strict inclusion.

\subsection{At-Most-One and Exactly-One Functions}

In this paper we are interested in two special cases of cardinality
constraints represented by ``at most one'' and ``exactly one'' functions.
These functions differ only on the zero input.

\begin{definition}
   For every $n \ge 1$, the function \(\amo_n(x_1, \dots, x_n)\) (\emph{at most one}) is defined as follows:
   Given an assignment \(\alpha:\{x_1, \ldots, x_n\} \to \{0,1\}\), the
   value \(\amo_n(\alpha)\) is \(1\) if and only if there is at most
   one index \(i\in\{1, \dots, n\}\) for which \(\alpha(x_i)=1\).
\end{definition}

\begin{definition}
   For every $n \ge 1$, the function \(\exone_n(x_1, \dots, x_n)\) (\emph{exactly one}) is defined as follows:
   Given an assignment \(\alpha:\{x_1, \ldots, x_n\} \to \{0, 1\}\), the value
   \(\exone_n(\alpha)\) is \(1\) if and only if there is exactly
   one index \(i\in\{1, \dots, n\}\) for which \(\alpha(x_i)=1\).
\end{definition}

We study propagation complete encodings of these two functions
using their common generalization called \penc{}
introduced in Definition~\ref{def:p-enc}.

\subsection{CNF Encoding}
\label{sub:cnf-enc} % chktex 24

We work with formulas in conjunctive normal form (CNF formulas).
For a standard notation see e.g.~\cite{CH11}. Namely, a \emph{literal}
is a variable \(x\) (\emph{positive literal}), or its negation \(\neg
   x\) (\emph{negative literal}).
If $x$ is a variable, then let $\lit{x}=\{x, \neg x\}$.
If $\vek{z}$ is a vector of variables, then we denote by $\lit{\vek{z}}$
the union of $\lit{x}$ over $x \in \vek{z}$.
For simplicity, we write \(x\in\vek{z}\) if \(x\) is a variable that
occurs in \(\vek{z}\), so \(\vek{z}\) is considered as a set here,
although, the order of the variables in \(\vek{z}\) is important.
Given a literal \(g\), the term
\(\var{g}\) denotes the variable in the literal \(g\), that is,
\(\var{g}=x\) for \(g \in\lit{x}\). Given a set of literals \(C\),
\(\var{C}=\{\var{g} \mid g\in C\} \).

A \emph{clause} is a disjunction of
a set of literals which does not contain a complementary pair of
literals. A formula is in \emph{conjunctive normal form} (\emph{CNF})
if it is a conjunction of a set of clauses. In this paper, we consider only
formulas in conjunctive normal form and we often simply refer to
a \emph{formula}, by which we mean a CNF formula.
We treat clauses as sets of
literals and formulas as sets of clauses. In particular, the order
of the literals in a clause or clauses in a formula is not important and we
use common set relations and operations (set membership, inclusion, set difference, etc.)
on clauses and formulas.
The empty clause (the contradiction) is denoted \(\bot\) and
the empty formula (the tautology) is denoted \(\top\).

A \emph{unit clause} consists of a single literal. A \emph{binary}
clause consists of two literals. A CNF formula, each clause of which
contains at most \(k\) literals, is said to be a \(k\)-CNF formula.

A \emph{partial assignment} \(\partass\) of variables \(\vek{z}\) is
a subset of \(\lit{\vek{z}}\) that does not contain a complementary pair
of literals, so we have \(|\partass\cap\lit{x}|\le 1\) for each \(x\in\vek{z}\).
By \(\phi(\partass)\) we denote the
formula obtained from $\phi$ by the partial setting of the variables
defined by $\partass$.

A CNF formula \(\varphi(\vek{z})\) represents a boolean function \(f\)
on the variables in $\vek{z}$. We say that a clause \(C\) is an
\emph{implicate} of a formula \(\phi\) if any satisfying assignment
\(\alpha\) of \(\phi\) satisfies \(C\) as well, i.e.
\(\phi(\alpha)=1\) implies \(C(\alpha)=1\) for every assignment
\(\alpha\). We denote this property with \(\phi\models C\). We say
that \(C\) is a \emph{prime implicate} of \(\phi\) if none \(C'\subset
   C\) is an implicate of \(\phi\). Note that whether a clause \(C\)
is a (prime) implicate of \(\phi\) depends only on the function \(f\)
represented by \(\phi\) and we can therefore speak about implicates of
\(f\) as well. We say that CNF \(\phi\) is \emph{prime} if it consists
only of prime implicates of \(\phi\). By the \emph{size} of the
formula \(\varphi\) we mean the number of clauses in
\(\varphi\), it is denoted as \(|\varphi|\) which is consistent with
considering a CNF formula as a set of clauses.

In this paper we also consider encodings of boolean functions defined
as follows.

\begin{definition}[Encoding]
   \label{def:cnf-enc}
   Let \(f(\vek{x})\) be a boolean function on variables \(\vek{x}=(x_1, \dots,
      x_n)\). Let \(\phi(\vek{x},\vek{y})\) be a CNF formula
   on \(n+\ell\) variables, where \(\vek{y}=(y_1, \dots,
      y_\ell)\).
   We call \(\phi\) an \emph{encoding} of \(f\) if
   for every \(\alpha\in{\{0, 1\}}^{n}\) we have
   \begin{equation}
      \label{eq:enc-def}
      f(\alpha)=1\Longleftrightarrow(\exists
      \beta\in{\{0, 1\}}^{\ell})\, \phi(\alpha, \beta)=1 \text{.}
   \end{equation}
   The variables
   in \(\vek{x}\) and \(\vek{y}\) are called \emph{input variables} and
   \emph{auxiliary variables}, respectively.
\end{definition}

\subsection{Propagation Complete Encoding}
\label{sub:pc-enc} % chktex 24

We are interested in encodings which are propagation
complete.  This notion relies on unit resolution which is a special
case of general resolution.
We say that two clauses \(C_1, C_2\) are resolvable, if there is
exactly one literal \(l\) such that \(l\in C_1\) and \(\neg l\in
   C_2\). The \emph{resolvent} of these clauses is then defined as
\(\resolvent{C_1}{C_2}=(C_1\cup C_2)\setminus\{l, \neg l\}\). If one
of \(C_1\) and \(C_2\) is a unit clause, we say that
\(\resolvent{C_1}{C_2}\) is derived by \emph{unit resolution} from
\(C_1\) and \(C_2\).
We say that a clause
\(C\) can be derived from \(\phi\) by \emph{unit resolution} (or
\emph{unit propagation}), if \(C\) can be derived from \(\phi\) by a
series of unit resolutions. We denote this fact with
\(\phi\vdash_1C\).

\begin{definition}[Propagation complete encoding]
   \label{def:pc-enc}
   Let \(f(\vek{x})\) be a boolean function on variables \(\vek{x}=(x_1, \dots,
      x_n)\). Let \(\phi(\vek{x},\vek{y})\) be a CNF formula
   on \(n+\ell\) variables, where \(\vek{y}=(y_1, \dots,
      y_\ell)\).
   We call \(\phi\) a \emph{propagation complete encoding} (PC encoding) of $f(\vek{x})$
if it is an encoding of \(f\) and for any \(g_1, \dots, g_p\in\lit{\vek{x}}\), \(p\ge 1\)
and for each \(h\in\lit{\vek{x}}\), such that
\begin{equation}
\label{eq:pc-enc-1}
f(\vek{x})\wedge\bigwedge_{i=1}^p g_i\models h
\end{equation}
we have
\begin{equation}
   \label{eq:pc-enc-2}
   \phi\wedge\bigwedge_{i=1}^p g_i\vdash_1
   h\hspace{1em}\text{or}\hspace{1em}
   \phi\wedge\bigwedge_{i=1}^p g_i\vdash_1 \bot\text{.}
\end{equation}
   If \(\phi\) is a prime CNF formula, we call it \emph{prime PC encoding}.
\end{definition}

Note that the definition of a propagation complete encoding is less
restrictive than requiring that formula \(\phi\) is propagation complete as
defined in~\cite{BM12}. The difference is that in a PC encoding we
only consider literals on input variables as assumptions and
consequences of unit propagation.
The definition of a propagation complete formula~\cite{BM12}
does not distinguish input and
auxiliary variables and the implication from~\eqref{eq:pc-enc-1}
to~\eqref{eq:pc-enc-2}
is required for all literals on all variables.

The following notation is used in the rest of the paper.
Let \(g_1, \dots, g_k\) be literals
on variables in \(\phi\).
Then \(\upclos[\phi]{g_1, \dots, g_k}\) denotes the set of
literals that can be derived by unit resolution from
\(\phi\wedge g_1\wedge\dots\wedge g_k\)
that is
\[\upclos[\phi]{g_1, \dots,
   g_k}=\left\{h \mid \phi\wedge\bigwedge_{i=1}^k g_i\vdash_1
   h\right\} \ .\]

\subsection{Propagation Complete Encodings of \(\amo_n\) and \(\exone_n\)}
\label{sub:pc-enc-amo} % chktex 24

Propagation complete encodings of \(\amo_n\) and \(\exone_n\) share
two common properties which we capture under the notion of \penc{}.

\begin{definition}[\penc{}]
   \label{def:p-enc}
   Let \(\phi(\vek{x}, \vek{y})\) be a formula with
   \(\vek{x}=(x_1, \dots, x_n)\), \(\vek{y}=(y_1, \dots, y_\ell)\),
   \(n\geq 1\), \(\ell\geq 0\). We say that \(\phi\) is a \penc{} if
   it satisfies the following two conditions.
   \begin{amocond}
      \item \(\phi\wedge x_i\) is satisfiable for each \(i\in\{1,
               \dots, n\}\),\label{cond:pc-amo-1}
      \item \(\phi\wedge x_i\vdash_1\neg x_j\) holds for each \(i,
            j\in\{1, \dots, n\}\) with \(i \neq j\),\label{cond:pc-amo-2}
   \end{amocond}
\end{definition}

One can easily verify that being a \penc{} is a necessary condition
for a formula to be a propagation complete encoding of \(\amo_n\) or \(\exone_n\).

\begin{lemma}
   \label{lem:p-enc}
   Let \(\varphi(\vek{x}, \vek{y})\) be a PC encoding of \(\amo_n\)
   or \(\exone_n\). Then \(\varphi\) is a \penc{}.
\end{lemma}

It turns out that conditions~\ref{cond:pc-amo-1}
and~\ref{cond:pc-amo-2} are enough to show the lower bounds
on the size of PC encodings of \(\amo_n\) and \(\exone_n\)
and these lower bounds are
derived in the following sections by proving a lower bound on
the size of \pencs{}.

Every \penc{} with \(n\) input variables is an encoding
of either \(\amo_n\) or \(\exone_n\), however, if it is
an encoding of $\exone_n$, it may
not be propagation complete. In particular, this means
that the converse of Lemma~\ref{lem:p-enc} is not true,
see Section~\ref{sec:relations-p-enc} for more detail.

\subsection{The Main Result}
\label{sub:main-result} % chktex 24

Let us introduce the following notation.

\begin{definition} \label{def:size-functions}
We denote the minimum size of a \penc{} with $n$ input variables
by \(\pencS{n}\) and the minimum size of a 2-CNF \penc{} with $n$
input variables by \(\pencQS{n}\).
\end{definition}

We pay special attention to 2-CNF encodings of \(\amo_n\).
The minimum size of these encodings is $\pencQS{n}$ as explained in
Section~\ref{sec:quad-lb}.
One can prove by contradiction that there are no 2-CNF encodings
of \(\exone_n\) of \(n\geq 3\) input variables as follows. Given an encoding
\(\varphi(\vek{x}, \vek{y})\) of \(\exone_n\),
we can eliminate an auxiliary variable $y$ from \(\varphi\)
by removing clauses containing $y$ or $\neg y$ and replacing
them with the resolvents of all pairs of these clauses resolvable using the
variable $y$. We call this step \emph{DP-elimination} of $y$,
since its repetition for all variables is one of the parts of
Davis-Putnam algorithm~\cite{DP60} (see also~\cite{BHMW09}).
After eliminating all auxiliary
variables, the remaining formula is
a 2-CNF representation of \(\exone_n\), since 2-CNF formulas
are closed under resolution. This is a contradiction, since
\(x_1\lor\cdots\lor x_n\) is a prime implicate of \(\exone_n\).

We are now ready to state the main result of this paper.

\begin{theorem}
   \label{thm:lower-bound}
Every PC encoding of $\amo_n$ or $\exone_n$ has size at least
$\pencS{n}$, the smallest size of a 2-CNF encoding of $\amo_n$
is equal to $\pencQS{n}$, and
   \begin{enumerate}
      \item\label{enum:lower-bound:1} For \(3\leq n \leq 8\) we have \(\pencS{n}=\pencQS{n}=3n-6\).
      \item\label{enum:lower-bound:2} For \(n \geq 9\) we have
         \(\pencS{n}\geq 2n+\sqrt{n}-2\).
      \item\label{enum:lower-bound:3} For \(n\geq 9\) we have
         \(\pencQS{n} \geq 2n+2\sqrt{n}-3\).
   \end{enumerate}
\end{theorem}

The lower bound for $3\leq n \leq 8$ is tight, since for every $n \ge 3$,
there is a 2-CNF encoding of $\amo_n$ of size $3n-6$ and it is a \penc{}.
The lower bound for $n \ge 9$ is almost tight, since
for every sufficiently large $n$, the product encoding~\cite{amoCHE_2010}
of $\amo_n$ has size $2n + 4 \sqrt{n} + O(n^{1/4})$ and is
a 2-CNF \penc{}.
Moreover, in Section~\ref{sec:relations-p-enc}, we prove
that $\pencS{n}$ is a close estimate of the minimum size of PC encodings of the
functions $\amo_n$ and $\exone_n$. Namely, for each of these functions,
there is a PC encoding of size at most $\pencS{n}+1$.

The first part of Theorem~\ref{thm:lower-bound} follows from
Lemma~\ref{lem:p-enc} and Lemma~\ref{lem:minimum-2-CNF-AMO}.
Our proof of parts~\ref{enum:lower-bound:2}
and~\ref{enum:lower-bound:3} relies on the notion of \pencs{} in
regular form we define below.
Let \(\varphi(\vek{x}, \vek{y})\) be a \penc{} with input variables
\(x_1, \ldots, x_n\).
Given a variable \(x_i\), \(i=1, \ldots, n\),
unit propagation on formula \(\phi\wedge x_i\) starts with clauses
which contain the negative literal \(\neg x_i\). It is
important to distinguish different types of \pencs{} according
to the structure of these clauses.
For each \(i=1, \dots, n\) let us denote
\begin{equation}
   \label{eq:define-sets-Q}
   \QA{i}=\{C\in\phi\;|\;\neg x_i\in C\}\text{.}
\end{equation}

\begin{definition}
   \label{def:regular-form}
   Let \(\varphi(\vek{x}, \vek{y})\) be a \penc{} with input variables
   \(x_1, \ldots, x_n\). We say that
   \(\phi(\vek{x}, \vek{y})\) is in \emph{regular form} if
   the following holds for each \(i\in\{1,\ldots,n\}\):
   \begin{regcond}
   \item\label{regular-cond1}\(\left|\QA[\phi]{i}\right|=2\).
   \item\label{regular-cond2}Clauses in \(\QA[\phi]{i}\) contain no input variables other than \(x_i\).
   \item\label{regular-cond3}Clauses in \(\QA[\phi]{i}\) are binary.
   \end{regcond}
\end{definition}

It is interesting to note that the construction of the product
encoding introduced by Chen~\cite{amoCHE_2010} leads to an encoding in
regular form and this form is probably the best for most values of $n$.
On the other hand, there are infinitely many rare values of $n$,
for which using a \penc{} not in regular
form allows to slightly reduce the size. This is used to describe the
product encoding in Section~\ref{sec:product-enc}.

The following theorem is used later to reduce the analysis of the
minimum size \pencs{} to the analysis of \pencs{} in regular form
and an induction argument. The theorem will be used for both
general CNF and 2-CNF formulas. Since the minimum size of a \penc{} can
be different in these two classes of formulas, we do not use the
assumption that $\phi(\vek{x}, \vek{y})$ is a minimum size \penc{}
and include condition (a) that
has the same effect and can be used for both general CNF and 2-CNF
formulas.

\begin{theorem}
   \label{thm:structure-general-CNF-2}
   If $\phi(\vek{x}, \vek{y})$ is a prime \penc{} with input variables
   \(x_1, \ldots, x_n\), $n \ge 4$,
   then at least one of the following holds:
   \begin{enumerate}[label={(\alph*)}] %%[label=(\roman*)]
      \item\label{enum:structure-general-CNF-2:1}
         There is a \penc{} $\phi'$ with \(n\) input variables, such
         that $|\phi| \ge |\phi'| + 1$.
      \item\label{enum:structure-general-CNF-2:2}
         There is a \penc{} $\phi'$ with \(n-1\) input variables, such
         that $|\phi| \ge |\phi'| + 3$.
      \item\label{enum:structure-general-CNF-2:3}
         Formula $\phi$ is in regular form.
   \end{enumerate}
   Moreover if \(\phi\) is a
   2-CNF formula, then so is \(\phi'\) in
   cases~\ref{enum:structure-general-CNF-2:1}
   and~\ref{enum:structure-general-CNF-2:2}.
\end{theorem}

We give a proof of Theorem~\ref{thm:structure-general-CNF-2}
in Section~\ref{sec:regular-form}. This theorem allows the following
approach to proving a lower bound. If a given \penc{} \(\varphi(\vek{x}, \vek{y})\)
is not in regular form, we use induction on \(n\), and if it is
in regular form, we prove a lower bound directly.
Although we combine Theorem~\ref{thm:structure-general-CNF-2} later
with additional arguments to prove stronger lower bounds,
the following simple corollary of this theorem
already gives the main term of the lower bound.

\begin{corollary}
   \label{cor:two-n}
   Let \(\varphi(\vek{x}, \vek{y})\) be a \penc{} with input
   variables \(x_1, \ldots, x_n\), \(n\geq 6\). Then \(\varphi\)
   consists of at least \(\pencS{n} \ge 2n\) clauses.
\end{corollary}
\begin{proof}
Assume that $\phi$ is a minimum size \penc{} for $n \ge 3$ input variables.
Without loss of generality, we can assume that it is
a prime \penc{} (see also Lemma~\ref{lem:prime-encoding} below).
We prove by induction on \(n\) that $|\phi| \ge \min(2n, 3n-6)$.
For $n \ge 6$, this implies the stated lower bound.
For \(n=3\), one can check
(see also Lemma~\ref{lem:base-cases} below) that
   \(\varphi\) must contain at least three clauses, thus
   \(|\varphi|\geq 3 = \min(2n, 3n-6)\). Now let us assume that \(n\geq 4\).
   Since \(\varphi\) is of minimum size,
no \penc{} with $n$ input variables with fewer clauses exists and
   item~\ref{enum:structure-general-CNF-2:1} of
   Theorem~\ref{thm:structure-general-CNF-2} does not apply. If
   \(\varphi\) is not in regular form then by
   item~\ref{enum:structure-general-CNF-2:2}
of Theorem~\ref{thm:structure-general-CNF-2},
there is a
   \penc{} \(\varphi'\) with \(n-1\) input variables such that
   \(|\varphi|\geq|\varphi'|+3\). By induction hypothesis
   \(|\varphi'|\geq \min(2(n-1), 3(n-1)-6) \ge \min(2n, 3n-6) - 3\), thus
   \begin{equation*}
      |\varphi|\geq |\varphi'|+3 \ge \min(2n, 3n-6)\text{.}
   \end{equation*}
If \(\varphi\) is in regular form, then we have
$|\varphi| \ge 2n \ge \min(2n, 3n-6)$, since by definition,
the union of $\QA{i}$, $i=1,\ldots,n$, contains $2n$ clauses.
\end{proof}

In order to prove the lower bounds presented in
Theorem~\ref{thm:lower-bound}, we perform a more careful analysis of
\pencs{} in regular form. In particular, in Section~\ref{sec:gen-lb}
we show the lower bound on \(\pencS{n}\) and in
Section~\ref{sec:quad-lb} we show the lower bound on \(\pencQS{n}\).

   \section{Known and Auxiliary Results}
\label{sec:def} % chktex 24

In this section we state known and preliminary results used throughout
the paper. We start by recalling some of the known good encodings
of $\amo_n$ with some modifications.

\subsection{Sequential Encoding}
\label{ssec:sequential-enc}

Let us present a variant of
the \emph{sequential encoding}~\cite{ccSIN_2005}, which addresses also more general
cardinality constraints. This construction has also been called
\emph{ladder encoding} in~\cite{amoHOL_2013}. The following recurrence
describes the sequential encoding $\phi^s_n$ of $\amo_n$ with a minor
simplification which reduces its size to $3n-6$ and the number
of auxiliary variables to $n-3$. The base case is
$$
\phi^s_3(x_1,x_2,x_3) = (\neg x_1 \vee \neg x_2) \wedge
(\neg x_1 \vee \neg x_3) \wedge (\neg x_2 \vee \neg x_3)
$$
and for each $n > 3$, let
$$
\phi^s_n(x_1, \ldots, x_n) =
(\neg x_1 \vee \neg x_2) \wedge (\neg x_1 \vee y) \wedge (\neg x_2 \vee y) \wedge
\phi^s_{n-1}(y, x_3, \ldots, x_n)
$$
where $y$ is an auxiliary variable not used in $\phi^s_{n-1}$.
By induction on $n$, one can verify that $\phi^s_n$ is an encoding of
\(\amo_n\) with $n-3$ auxiliary variables and of size $|\phi^s_n|=3n-6$.
Since it is a prime 2-CNF\@, it is propagation
complete, see~\cite{BBCGK13}. Hence, we have the following.

\begin{lemma} \label{lem:simple-upper-bound}
   For every $n \ge 3$, there is a 2-CNF PC encoding of \(\amo_n\) of
   size \(3n-6\).
\end{lemma}

Since $\amo_{n-1}$ is a symmetric function, the order of the variables
in the formula for this function can be chosen arbitrarily without changing
the function. When a different order
of the variables is used in a recurrence, the obtained formula
has a different form. Let us introduce the \emph{tree encoding}
by the following recurrence. The base case is
$$
\phi^t_3(x_1,x_2,x_3) = \phi^s_3(x_1,x_2,x_3)
$$
and for each $n > 3$, let
$$
\phi^t_n(x_1, \ldots, x_n) =
(\neg x_1 \vee \neg x_2) \wedge (\neg x_1 \vee y) \wedge (\neg x_2 \vee y) \wedge
\phi^t_{n-1}(x_3, \ldots, x_n, y)
$$
where $y$ is an auxiliary variable not used in $\phi^t_{n-1}$.
By a similar argument as above, $\phi^t_n$ is a PC encoding of $\amo_n$.

The size of the formulas $\phi^s_n$ and $\phi^t_n$ is the same
and both are 2-CNF formulas. Let us consider a graph, whose vertices
are variables and edges are the two-element sets $\var{C}$, where $C$ is a clause
in the formula. Both the formulas $\phi^s_n$ and $\phi^t_n$
can be decomposed into triples of clauses which correspond to triangles
in their graph and the triangles are connected via their vertices
into a tree structure. In the graph for
$\phi^s_n$, these triangles form a path of length $n-3$ and in the graph
for $\phi^t_n$, they form a tree with diameter $O(\log n)$.

\subsection{Product Encoding}
\label{sec:product-enc}

Chen~\cite{amoCHE_2010} introduced the \emph{product
   encoding} of \(\amo_n\) which has size \(2n+4\sqrt{n} +
   O(\sqrt[4]{n})\). It turns out that $n=25$ is the smallest
number of the variables, for which
the product encoding outperforms the sequential encoding (\(68\) vs.
\(69\) clauses). On the other hand, we show below that the sequential
encoding is the smallest possible for \(n\le 8\). It is not clear
whether this holds also for \(9 \le n \le 24\).

Let us present a slightly optimized version of the product encoding
using a combination with sequential encoding for some values of $n$.
The combination is obtained by considering two candidates
for the recursive construction of the product encoding $\phi^p_n$
and using the better of them for each $n \ge 7$. The base case
for $n=3$ is
$$
\phi^p_3(x_1,x_2,x_3) = \phi^s_3(x_1,x_2,x_3)
$$
If $n \ge 4$, the first candidate formula for $\phi^p_n(\vek{x})$ is
\begin{equation} \label{eq:construct-by-reduce}
(\neg x_1 \vee \neg x_2) \wedge (\neg x_1 \vee y) \wedge (\neg x_2 \vee y) \wedge
\phi^p_{n-1}(y, x_3, \ldots, x_n)
\end{equation}
as in the recurrence used for the sequential encoding.
If $n \le 6$, let $\phi^p_n(\vek{x})$ be~\eqref{eq:construct-by-reduce}.
If $n \ge 7$, we use the formula described by Chen~\cite{amoCHE_2010}.
Let $m_1 = \lceil \sqrt{n}\, \rceil$ and $m_2 = \lceil n/m_1 \rceil$.
Clearly, we have $m_1 \ge m_2 \ge 3$.
Arrange the input variables
in $n$ pairwise different cells of a rectangular array of dimension $m_1 \times m_2$.
Let $r:\{1,\ldots,n\} \to \{1,\ldots,m_1\}$
and $c:\{1,\ldots,n\} \to \{1,\ldots,m_2\}$ be the functions, such
that $r(i)$ is the row index and $c(i)$ the column index of the
cell containing $x_i$. Let $y_j$, $j=1,\ldots,m_1$ and
$z_j$, $j=1,\ldots,m_2$ be new auxiliary variables. Then, the second
candidate for $\phi^p_n(\vek{x})$ is the formula
\begin{equation} \label{eq:product-encoding}
   \bigwedge_{i=1}^n (\neg x_i \vee y_{r(i)})
\wedge \bigwedge_{i=1}^n (\neg x_i \vee z_{c(i)})
\wedge \phi^p_{m_1}(\vek{y}) \wedge \phi^p_{m_2}(\vek{z})
\end{equation}
It is worth noting that formula~\eqref{eq:product-encoding} is in
regular form, see Definition~\ref{def:regular-form}.
The size of~\eqref{eq:construct-by-reduce} is $3 + |\phi^p_{n-1}|$
and the size of~\eqref{eq:product-encoding}
is $2n + |\phi^p_{m_1}| + |\phi^p_{m_2}|$.
Let $\phi^p_n(\vek{x})$ be the smaller of these formulas,
where any of the candidates can be used, if their sizes are the same.
It appears that formula~\eqref{eq:construct-by-reduce} is smaller
than~\eqref{eq:product-encoding} for $n \le 23$ and for infinitely
many other numbers, in particular, for the numbers $n=m^2+1$ and
$n=m(m+1)+1$, where $m \ge 5$ is an integer. This can be explained
as follows. If $n=m^2$ or $n=m(m+1)$, then
$\phi^p_n(\vek{x})$ is given by~\eqref{eq:product-encoding}.
In this case, the size of the candidate for $\phi^p_{n+1}(\vek{x})$ given
by~\eqref{eq:product-encoding} is at least
$|\phi^p_n(\vek{x})|+4$ and the size of the candidate for $\phi^p_{n+1}(\vek{x})$
given by~\eqref{eq:construct-by-reduce} is $|\phi^p_n(\vek{x})|+3$.

Clearly, the size of $\phi^p_n$ is at most $3 + |\phi^p_{n-1}|$
which is the size of~\eqref{eq:construct-by-reduce} and
using this, one can prove by induction on $n$ that for all \(n \ge 3\),
we have
\begin{equation}
   \label{eq:phi-p-bound}
   |\phi^p_n| \le 3n-6.
\end{equation}
Asymptotically, a better bound was proven
by Chen~\cite{amoCHE_2010}. We present here a proof of this bound for
completeness.

\begin{lemma}[Chen~\cite{amoCHE_2010}] \label{lem:product-encoding}
We have $|\phi^p_n| = 2n + 4\,\sqrt{n} + O(\sqrt[4]{n})$.
\end{lemma}

\begin{proof}
If $n \ge 7$ and $m=\lceil \sqrt{n}\, \rceil$, the size
of the product encoding satisfies
\begin{equation} \label{eq:size-product-enc}
   |\phi^p_n| \le 2n + 2|\phi^p_m|\text{.}
\end{equation}
By~\eqref{eq:phi-p-bound}, we have $|\phi^p_m| \le 3m-6$.
Together with~\eqref{eq:size-product-enc} we obtain
$$
|\phi^p_n| \le 2n + O(\sqrt{n})\text{.}
$$
Using this bound for $|\phi^p_m|$ and using~\eqref{eq:size-product-enc}
again, we obtain
$$
|\phi^p_n| \le 2n + 4 \sqrt{n} + O(\sqrt[4]{n})
$$
as required.
\end{proof}

\subsection{Relationship to Monotone Circuits}
\label{sec:monotone-circuits} % chktex 24

Let us briefly describe a connection between the product encoding and
the monotone circuit of the size $kn + o(n)$ for the function $T^n_k$,
described in~\cite{D84} and in Section 6, Theorem {2.3} in~\cite{W87}.
If $k=2$, the construction yields the product encoding for $\amo_n$.
More generally, if $k$ is any constant, we obtain a PC encoding of
``at most $(k-1)$'' of size $kn + o(n)$. This is the smallest
known encoding for this constraint, if $k \le 16$. On the other
hand, for every $k \ge 32$, a smaller encoding can be obtained using
Batcher's sorting network. For large $k$, an even smaller
encoding of size
$C_k n + O(1)$, where $C_k =O(\log k)$, can be obtained using
AKS sorting networks.

Let $T^n_k$ be the threshold function ``at least $k$''
of $n$ variables. By the results cited above, there is a monotone
circuit for this function consisting of $kn + o(n)$
binary AND and OR gates. In order to obtain asymptotically
the same number of clauses in an encoding, the circuit has to
be transformed in such a way that we replace groups of binary OR gates
computing a disjunction of several previous gates by a single OR gate
with multiple inputs. The Horn part of the Tseitin encoding of
the circuit after this transformation consists of $kn + o(n)$ clauses.
If we add a negative unit clause on the output of the circuit,
we obtain an encoding of the ``at most $(k-1)$'' constraint.
Moreover, using the specific structure of the
circuit, one can verify that this encoding is propagation complete.
In particular, if $k=2$, the obtained encoding is the product
encoding of the constraint ``at most one''.

\subsection{P-Encodings and Encodings of $\amo_n$ and $\exone_n$}
\label{sec:relations-p-enc} % chktex 24

We use \pencs{} as a representation of common properties
of PC encodings of $\amo_n$ and $\exone_n$. Although
the converse of Lemma~\ref{lem:p-enc} is not true (see
below for an example) a partial converse is valid.

\begin{lemma}
   \label{lem:p-enc-is-enc}
   A \penc{} with \(n\) input variables is an encoding (not
   necessarily propagation complete) of either
   \(\amo_n\) or \(\exone_n\).
\end{lemma}
\begin{proof}
   Consider a \penc{} \(\varphi(\vek{x}, \vek{y})\) where
   \(\vek{x}=(x_1, \ldots, x_n)\).
The functions $\amo_n$ and $\exone_n$ differ only on the zero
assignment. In order to prove the statement, it is sufficient
to prove that the function encoded by $\varphi(\vek{x}, \vek{y})$
agrees with $\amo_n$ and $\exone_n$ on the remaining assignments.

Consider a non-zero assignment \(\alpha\) of the input variables.
First,
   assume that \(\alpha(x_i)=\alpha(x_j)=1\) for two indices \(1\leq
      i<j\leq n\). Such \(\alpha\) is a falsifying assignment of both
   functions \(\amo_n\) and \(\exone_n\). By
   condition~\ref{cond:pc-amo-2} \(\alpha\) cannot be extended to a
   model of \(\varphi\). For the remaining case assume that
   \(\alpha(x_i)=1\) for some \(i\in\{1, \ldots, n\}\) and
   \(\alpha(x_j)=0\) for all \(j\in\{1, \ldots, n\}\setminus\{i\}\).
Such $\alpha$ is a satisfying assignment of both the functions
$\amo_n$ and $\exone_n$.
By condition~\ref{cond:pc-amo-1} we have that \(\varphi\land
      x_i\) is satisfiable and condition~\ref{cond:pc-amo-2} implies
   that any satisfying assignment of \(\varphi\land x_i\) sets all
   other input variables to \(0\). This means that \(\alpha\) can be
   extended to a satisfying assignment of \(\varphi\).
\end{proof}

Consider the formula
\[
\phi'=
(x_1 \vee \ldots \vee x_{n-2}  \vee x_n \vee y) \wedge
(x_{n-1} \vee x_n \vee \neg y) \wedge
\phi(\vek{x})
\]
where $\phi(\vek{x})$ is the prime representation of $\amo_n$
and $y$ is an auxiliary variable. One can verify that this
formula is a \penc{}, is an encoding of $\exone_n$, however,
is not a PC encoding of $\exone_n$, since
\(\phi' \wedge \bigwedge_{j\in\{1, \ldots, n-1\}}\neg x_j\not\vdash_1 x_n\).
This implies that the converse of Lemma~\ref{lem:p-enc} is not true.

Although we believe that the size of the
smallest PC encoding of $\amo_n$ is $\pencS{n}$ and the
size of the smallest PC encoding of $\exone_n$ is $\pencS{n}+1$,
we can prove only the following bounds.

\begin{proposition}\label{prop:relations-for-sizes}
Let \(n\geq 2\), let \(\varphi_1(\vek{x}, \vek{y})\) be a smallest
PC encoding of \(\amo_n(\vek{x})\) and let \(\varphi_2(\vek{x},
   \vek{z})\) be a smallest PC encoding of
\(\exone_n(\vek{x})\). Then
$$
\begin{array}{c}
\pencS{n} \le |\varphi_1| \le \pencS{n}+1\\
\pencS{n} \le |\varphi_2| \le \pencS{n}+1
\end{array}
$$
\end{proposition}
\begin{proof}
The lower bounds follow from Lemma~\ref{lem:p-enc}.
Let $\phi(\vek{x}, \vek{y})$ be a \penc{} of size $\pencS{n}$.
One can verify that
\begin{equation}
      \label{eq:strange-PC-enc-AMO} (\neg x_1 \vee z) \wedge \phi(z,
      x_2, \ldots, x_n, \vek{y})
\end{equation}
where $z$ is a new auxiliary variable, is a PC encoding
of $\amo_n(\vek{x})$. This implies \(|\varphi_1|\leq \pencS{n}+1\).
Moreover, one can verify that
$$
(x_1 \vee \ldots \vee x_n) \wedge \phi(\vek{x}, \vek{y})
$$
is a PC encoding of \(\exone_n(\vek{x})\). This
implies \(|\varphi_2|\leq \pencS{n}+1\).
\end{proof}

\subsection{Simple Reductions of Encodings}
\label{sub:simple-reductions} % chktex 24

In this section we present additional properties of encodings that
can be assumed without loss of generality, since every encoding can
be modified to satisfy them without increase of the size.

\begin{lemma} \label{lem:prime-encoding}
The prime CNF formula obtained from a given \penc{} $\phi(\vek{x}, \vek{y})$
by replacing every clause by a prime implicate contained in it,
is also a \penc{}.
\end{lemma}

\begin{proof}
Consider a clause $C$ in the original formula and a prime implicate $C'$
contained in it. Replacing $C$ by $C'$ does not change the function
represented by the formula and one can verify that
the conditions~\ref{cond:pc-amo-1} and~\ref{cond:pc-amo-2}
remain satisfied.
Repeating this for all clauses of $\phi(\vek{x}, \vek{y})$
proves the lemma.
\end{proof}

If the number of occurrences of an auxiliary variable \(y\) in an
encoding \(\phi\) is at most 4, then DP-elimination of \(y\) does not
increase the size of the formula (see Section~\ref{sub:main-result}
for definition of DP-elimination) and leads to an encoding of
the same function with a smaller number of auxiliary variables.
This allows us to make the following observation.

\begin{lemma} \label{lem:aux-var-occurrences}
Let \(\phi(\vek{x}, \vek{y})\) be an encoding of a function $f(\vek{x})$
of minimum size that, moreover, has the minimum number of auxiliary variables
among such encodings. Then any auxiliary variable
   \(y\in\vek{y}\) occurs in at least 5 clauses of \(\phi\).
\end{lemma}

With a little effort one can show that DP-elimination also
preserves propagation completeness of an encoding.
In particular, Lemma~\ref{lem:aux-var-occurrences} holds
also for a PC encoding of minimum size, however,
this is not used in this paper.

\subsection{Substituting Variables in Unit Propagation}

One of the reduction steps we use later to simplify an encoding
is the substitution of a variable with a literal on a variable already
present in the formula. If \(\phi(\vek{z})\) is a formula and
\(g_1,g_2 \in\lit{\vek{z}}\), we denote by \(\subst{\phi}{g_1}{g_2}\)
the formula obtained from \(\phi\) using the substitution
$g_1 \leftarrow g_2$. More precisely, if the literal $g_1$ is positive, then
the variable $\var{g_1}$ is substituted by the literal $g_2$. If $g_1$
is negative, then the variable $\var{g_1}$ is substituted by the
literal $\neg g_2$. An important property of this operation is that
this kind of substitution preserves unit propagation. 

\begin{lemma}
   \label{lem:identification-propagation}
   Let \(\phi(\vek{z})\) be a formula, let \(g_1,g_2,h \in\lit{\vek{z}}\),
   such that \(\var{g_1} \not= \var{h}\) and assume,
   \(\subst{\phi}{g_1}{g_2}\) is satisfiable. Then
   \[\phi \vdash_1 h \Longrightarrow
      \subst{\phi}{g_1}{g_2} \vdash_1 h \text{.}\]
\end{lemma}

Lemma~\ref{lem:identification-propagation} is a consequence of a more
general statement with an essentially the same proof which we are going to show first.
Let us consider a substitution \(t:\vek{z} \to \lit{\vek{z}}\)
of the variables in \(\vek{z}\) by literals on the same set of
the variables. The substitution extends to the literals so that for
every $x \in \vek{z}$, we have $t(\neg x) = \neg t(x)$.
Moreover, the substitution extends to the clauses and the formulas
in CNF as follows. If \(C=g_1\vee \cdots \vee g_k\) is
a clause with variables from \(\vek{z}\) then
\(t(C)\) is defined as \(t(g_1)\vee \cdots \vee t(g_k)\),
if there is no complementary pair of literals among
\(t(g_1), \ldots, t(g_k)\), and \(\top\) otherwise. If \(\phi\)
is a CNF formula, then \(t(\phi)=\bigwedge_{C\in\phi}t(C)\), where
\(t(\phi)=\top\) in case \(t(C)=\top\) for all \(C\in\phi\).
In particular,
\(\subst{\phi}{g_1}{g_2}=t(\phi)\) where \(t\) is
a map on the literals defined for every literal \(e\) as
\begin{equation}
   \label{eq:identification}
   t(e) =
   \begin{cases}
      e   & \text{if \(\var{e} \not= \var{g_1}\)}\\
      g_2 & \text{if \(e = g_1\)}\\
      \neg g_2 & \text{if \(e = \neg g_1\).}\\
   \end{cases}
\end{equation}
Applying a substitution to a formula preserves resolution proofs
as we show in the following lemma.

\begin{lemma} \label{lem:transform-general}
Let $\phi(\vek{z})$ be a formula on the variables $\vek{z}$
and let $C_1, \ldots, C_k$ be a resolution proof of $C_k$ from $\phi$.
If $t:\vek{z} \to \lit{\vek{z}}$ is a substitution as above,
then there is a sequence $D_1, \ldots, D_k$,
where each $D_i$ is a clause or $\top$,
such that the following implications are satisfied
\begin{align*}
C_i \in \phi & \Longrightarrow D_i = t(C_i)\\
C_i \not\in \phi & \Longrightarrow D_i \subseteq t(C_i) \not= \top
\mbox{\ or\ } D_i = t(C_i) = \top
\end{align*}
and the sequence of clauses contained in $D_1, \ldots, D_k$
is a resolution proof of the clauses contained in it from the clauses in $t(\phi)$.
If the original proof is a unit resolution proof,
so is the derived proof.
\end{lemma}

\begin{proof}
For each $i=1,\ldots,k$,
we have either $C_i \in\phi$ or $C_i=\resolvent{C_r}{C_s}$,
where $r, s < i$.
In order to prove the claim,
let us construct $D_i$ by induction on $i=1,\ldots,k$.
Some of the clauses can repeat in the constructed sequence.
Assume, the sequence $D_1,\ldots,D_{i-1}$ is constructed and
is empty or satisfies the requirements formulated above.
If $C_i \in \phi$, choose $D_i = t(C_i)$.
If $C_i \not\in \phi$, then
$C_i=\resolvent{C_r}{C_s}$, where $r, s < i$
and $C_r = g \vee A$, $C_s = \neg g \vee B$, and $C_i = A \vee B$
for some literal $g$ and sets of literals $A$ and $B$.
If $t(C_i) = t(A \vee B) = \top$, then choose $D_i=\top$.
Otherwise, there is no conflict in $t(A) \cup t(B)$.

If, moreover, the variable $\var{t(g)}$ has no occurence
in $t(A) \cup t(B)$, then $t(C_r)$ and $t(C_s)$ are clauses
and are resolvable using the literals $t(g)$ and $t(\neg g)$.
This implies $D_r \subseteq t(C_r)$, $D_s \subseteq t(C_s)$,
and $t(C_i)=\resolvent{t(C_r)}{t(C_s)}$.
If $D_r$ and $D_s$ are resolvable, then
$\resolvent{D_r}{D_s} \subseteq t(C_i)$ and we can
choose $D_i=\resolvent{D_r}{D_s}$. Otherwise,
either $t(g) \not\in D_r$ and we have $D_r \subseteq t(C_i)$
or $t(\neg g) \not\in D_s$ and we have $D_s \subseteq t(C_i)$.
Hence, we can choose $D_i=D_r$ or $D_i=D_s$ so that
$D_i \subseteq t(C_i)$.

Assume, some of the literals $t(g)$ and $t(\neg g)$
has an occurence in $t(A) \cup t(B)$. Since $t(A) \cup t(B)$
is a clause, only one of the literals $t(g)$ and $t(\neg g)$
is contained in it. If $t(g) \in t(A) \cup t(B)$, then
$t(C_r) = t(g) \cup t(A) \subseteq t(A) \cup t(B) = t(C_i)$
and we can choose $D_i=D_r \subseteq t(C_r) \subseteq t(C_i)$.
Similarly, if $t(\neg g) \in t(A) \cup t(B)$, we can
choose $D_i=D_s \subseteq t(C_s) \subseteq t(C_i)$.

If $C$ is a unit clause, then $t(C)$ is a unit clause.
This implies the last statement of the lemma.
\end{proof}

Lemma~\ref{lem:identification-propagation} now easily follows from
Lemma~\ref{lem:transform-general}.

\begin{proof}[Proof of Lemma~\ref{lem:identification-propagation}]
If \(t\) is defined as in~\eqref{eq:identification}, then
\(\subst{\phi}{g_1}{g_2}=t(\phi)\) and $t(h)=h$.
Lemma~\ref{lem:transform-general} implies
$t(\phi) \vdash_1 \bot$ or
$t(\phi) \vdash_1 t(h)$.
The first is excluded because
$\subst{\phi}{g_1}{g_2} = t(\phi)$
is satisfiable. Hence, we have
$\phi[g_1\leftarrow g_2] \vdash_1 h$
as required.
\end{proof}

   \section{Reducing to Regular Form}
\label{sec:regular-form} % chktex 24

This section is devoted to the proof of
Theorem~\ref{thm:structure-general-CNF-2}. We start with basic
properties of \pencs{}.

\begin{lemma}
   \label{lem:amo-pc}
   Let $n \ge 2$ and let \(\phi(\vek{x}, \vek{y})\) be a \penc{} with input
   variables \(\vek{x}=(x_1, \dots, x_n)\) and auxiliary
   variables \(\vek{y}=(y_1, \dots, y_\ell)\).
   For each distinct \(x_i, x_j \in \vek{x}\) it holds that
   \begin{enumerate}[label={(\alph*)}]
      \item   \(\phi\wedge x_i\not\vdash_1 x_j\),\label{enum:amo-pc:posToPos}
      \item   \(\phi\wedge \neg x_i\not\vdash_1 \neg x_j\),\label{enum:amo-pc:negToNeg}
      \item  \(\phi\) contains a binary clause containing the literal \(\neg x_i\).\label{enum:amo-pc:binClause}
   \end{enumerate}
\end{lemma}
\begin{proof}
   Suppose that \(\phi\) satisfies the assumption.
The claims of the lemma can be proven as follows.
   \begin{enumerate}[label={(\alph*)}]
      \item If $\phi \wedge x_i \vdash_1 x_j$, then
         $\phi\wedge x_i \vdash_1\bot$, since
         $\phi\wedge x_i \vdash_1 \neg x_j$ by condition~\ref{cond:pc-amo-2}. This contradicts
         condition~\ref{cond:pc-amo-1}.
      \item If \(\phi\wedge \neg x_i\vdash_1 \neg x_j\),
         then $\phi\wedge x_j \vdash_1 \bot$, since
         $\phi\wedge x_j \vdash_1 \neg x_i$ by condition~\ref{cond:pc-amo-2}. This contradicts
         condition~\ref{cond:pc-amo-1}.
      \item Since $\phi\wedge x_i \vdash_1 \neg x_j$, there is a
         series of unit resolutions starting from \(x_i\), whose
         first step uses a binary clause containing \(\neg x_i\).
         \qedhere
   \end{enumerate}
\end{proof}

The following lemma shows that fixing any set of input variables to
zero in a \penc{} gives us a \penc{} on the remaining input variables.

\begin{lemma}
   \label{lem:amos-fix}
   Let \(\phi(\vek{x}, \vek{y})\) be a \penc{}
   and let $I\subset\{1, \ldots, n\}$ be a non-empty set of indices
and consider the partial assignment
   \(\partass=\{\neg x_j \ | \ j \not\in I\} \).
   Then \(\phi(\partass)\) is a \penc{} with the input
   variables \(x_i\), \(i\in I\).
\end{lemma}
\begin{proof}
   If $i \in I$, then $\phi \wedge x_i$ derives all the
   literals $\neg x_k$ for $k \in \{1,\ldots,n\} \setminus \{i\}$
   including all the literals in $\partass$ and does not derive
   a contradiction. Hence, the propagation from
   $\phi(\partass) \wedge x_i$ cannot derive a contradiction
   and derives $\neg x_k$ for all $k \in I \setminus \{i\}$.
   It follows that $\phi(\partass)$ with the input variables $x_I$
   satisfies~\ref{cond:pc-amo-1} and~\ref{cond:pc-amo-2}.
\end{proof}

We now concentrate on clauses with negative literals on input
variables.

\begin{lemma}
   \label{lem:no-pos-and-neg-x}
   Let $\phi(\vek{x}, \vek{y})$ be a prime \penc{},
   $C \in \phi$, and $\neg x_i \in C$. Then
   one of the following is satisfied
   \begin{enumerate}[label={(\roman*)}]
      \item $C=\neg x_i \vee A$, where $\emptyset \not= A \subseteq \lit{\vek{y}}$,
      \item $C=\neg x_i \vee \neg x_j$ for some $j \not= i$.
   \end{enumerate}
\end{lemma}

\begin{proof}
   We have $C=\neg x_i \vee A$ for a non-empty set of literals $A$.
If $A$ contains a literal on an input variable, consider the following
two cases.
   \begin{itemize}
      \item If there is a literal \(\neg x_j\in A\) for some
         \(j\neq i\),
consider the clause $\neg x_i\vee \neg x_j$. This clause
is an implicate of $\phi(\vek{x}, \vek{y})$ due to
property~\ref{cond:pc-amo-2} and it is prime due to
property~\ref{cond:pc-amo-1}. Hence, necessarily
\(C=\neg x_i\vee \neg x_j\).
      \item If \(x_j\in A\) for some \(j\neq i\) then
         \(C \setminus \{x_j\} = \resolvent{C}{\neg x_i\vee\neg x_j}\) is an
         implicate as well which is in contradiction with primality of
         \(\phi\).
   \end{itemize}
Otherwise, $\emptyset \not= A \subseteq \lit{\vek{y}}$.
\end{proof}

The following proposition shows that \(|\QA{i}|\geq 2\) for every
input variable \(x_i\in\vek{x}\) in a minimum size \penc{}
\(\varphi(\vek{x}, \vek{y})\).

\begin{proposition}
   \label{lem:single-negative}
   Let \(n\geq 3\) and
   let \(\phi(\vek{x}, \vek{y})\) be a \penc{} with \(n\) input variables
   \(\vek{x}=(x_1, \ldots, x_n)\).
   Let \(x_i\in \vek{x}\) and suppose that  \(|\QA{i}|=1\).
   Then there is another \penc{} $\phi'$ with input variables
   \(\vek{x}\) and
   satisfying $|\phi| \ge |\phi'| + 1$.
   Moreover, if $\phi$ is a 2-CNF formula, then so is $\phi'$.
\end{proposition}
\begin{proof}
   Using Lemma~\ref{lem:prime-encoding}, we can assume that $\phi$ is a prime formula.
By Lemma~\ref{lem:amo-pc}, there is a binary clause \(C=\neg x_i\vee e\in \phi\)
with some \(e\in\lit{\vek{x}\cup\vek{y}}\). Let us assume for a contradiction
   that \(\var{e}=x_j\) with \(j\neq i\).
   By Lemma~\ref{lem:no-pos-and-neg-x}, \(C=\neg x_i\vee \neg x_j\).
   Let \(x_k\in\vek{x}\setminus\{x_i,x_j\}\).
   We have that \(\phi\wedge x_i\vdash_1\neg x_k\).
   Since $C$ is the only clause of $\phi$ containing \(\neg x_i\),
   unit resolution uses $x_i$ to derive $\neg x_j$ and does not use $x_i$
   in any of the later steps. Hence, we have
   \(\phi\wedge \neg x_j\vdash_1\neg x_k\), which is a contradiction
   with Lemma~\ref{lem:amo-pc}\ref{enum:amo-pc:negToNeg}.
   This implies \(e\in\lit{\vek{y}}\).

Consider the substitution $e \leftarrow x_i$ and
let us show that \(\phi'=\subst{\phi}{e}{x_i}\) satisfies the
   conditions~\ref{cond:pc-amo-1} and~\ref{cond:pc-amo-2}.
   \begin{amocond}
   \item
      Let us show that \(\phi'\wedge x_k\) is a
      satisfiable formula for each \(k\in\{1,\dots ,n\}\). If \(k=i\), we have that \(\phi\wedge x_i\)
      is satisfiable and that
      \(\phi\wedge x_i\vdash_1 e\) using the
      clause \(C\). Thus, the formula \(\phi\wedge x_i\wedge e\) is
satisfiable and both literals \(x_i\) and \(e\) get value $1$ in each of its
satisfying assignments. It follows that
\(\phi'\wedge x_i = (\phi\wedge x_i)[e \leftarrow x_i]\)
is satisfiable as well.

      If \(k\neq i\), we have \(\phi\wedge x_k\vdash_1\neg x_i\). Since
      \(C=\neg x_i\vee e\) is the only clause in
      \(\phi\) that contains \(\neg
         x_i\), it holds that \(\phi\wedge
         x_k\vdash_1\neg e\). Thus, both literals \(x_i\) and
      \(e\) get value $0$ in any
      satisfying assignment of \(\phi\wedge x_k\). It follows that
\(\phi'\wedge x_k = (\phi \wedge x_k)[e \leftarrow x_i]\)
is satisfiable as well.
   \item Follows from Lemma~\ref{lem:identification-propagation}
for the formula $\phi' \wedge x_k = (\phi \wedge x_k)[e \leftarrow x_i]$
and $h=\neg x_\ell$, where $k,\ell \in\{1,\dots ,n\}$, $k \not= \ell$.
   \end{amocond}
The substitution $e \leftarrow x_i$ changes \(C\) to $\top$
   which is omitted in \(\phi'\).
   Hence   $\phi'$ has size smaller than $\phi$. This completes the proof.
\end{proof}

Let us present an example of a formula which shows that
Proposition~\ref{lem:single-negative} does not hold for PC encodings of $\amo_n$
instead of \pencs{}.
The formula~\eqref{eq:strange-PC-enc-AMO}
with auxiliary variables $\vek{z}=(z, \vek{y})$
is a PC encoding of $\amo_n(\vek{x})$ with a single occurence of $\neg x_1$,
for which the construction in Proposition~\ref{lem:single-negative} provides
a formula $\phi'$ which is a PC encoding of $\exone_n(\vek{x})$.

\pencs{} that are not in regular form can be reduced to \pencs{}
with a smaller number of input variables by the following statements. We
start with \pencs{} which violate Condition~\ref{regular-cond1}
of Definition~\ref{def:regular-form}. Recall that
this condition requires that
\(|\QA{i}|=2\) for every input variable \(x_i\in\vek{x}\) of a \penc{}
in regular form.

\begin{proposition}
   \label{prop:number-neg-x-general-CNF}
   Let \(\phi(\vek{x}, \vek{y})\) be a \penc{} with \(n\geq 3\) input
   variables \(\vek{x}=(x_1, \ldots, x_n)\),
   such that \(\left|\QA[\varphi]{i}\right|\neq 2\) for some \(i\in\{1, \ldots,
         n\}\).
   Then, there is a formula $\phi'$, which satisfies one of the following
   \begin{enumerate} [label={(\alph*)}]
      \item\label{enum:number-neg-x-general-CNF:1}
         $\phi'$ is a \penc{} with \(n\) input variables and $|\phi| \ge |\phi'| + 1$,
      \item\label{enum:number-neg-x-general-CNF:2}
         $\phi'$ is a \penc{} with \(n-1\) input variables and $|\phi| \ge |\phi'| + 3$.
   \end{enumerate}
   Moreover, if $\phi$ is 2-CNF, then so is $\phi'$.
\end{proposition}

\begin{proof}
   Assume, $\left| \QA{i} \right| \not= 2$ for some $i\in\{1,\dots,n\}$.
   Lemma~\ref{lem:amo-pc} implies that \(|\QA{i}| \ge 1\).
   Assume, \(|\QA{i}|=1\).
   According to Proposition~\ref{lem:single-negative}, there is a \penc{}
   satisfying condition~\ref{enum:number-neg-x-general-CNF:1}
   of the conclusion.
   If $|\QA{i}| \ge 3$, then setting $x_i=0$
   yields a formula $\phi'$ of size at most $|\phi|-3$ and
   this formula is a \penc{} with \(n-1\) input variables by Lemma~\ref{lem:amos-fix}.
   Hence, condition~\ref{enum:number-neg-x-general-CNF:2}
   of the conclusion is satisfied.
In both cases, $\phi'$ is obtained from $\phi$ by a substitution that
does not increase the size of clauses, so also the last part of the statement
is satisfied.
\end{proof}

Now, we look at \pencs{} which satisfy
Condition~\ref{regular-cond1} but do not satisfy
Condition~\ref{regular-cond2} of Definition~\ref{def:regular-form}.
For this purpose, we use the following auxiliary lemma.

\begin{lemma}
   \label{lem:two-pure-input-clauses}
   Let $\phi(\vek{x}, \vek{y})$ be a \penc{} with \(n\geq 4\) input variables
   \(\vek{x}=(x_1, \ldots, x_n)\)
   and let \(i,j,k\in\{1,\dots,n\}\) be three different indices.
   Then \(\QA{i}\neq\{\neg x_i \vee \neg x_j,\, \neg x_i \vee \neg x_k\}\).
\end{lemma}

\begin{proof}
   Let $\ell \in \{1,\dots,n\} \setminus \{i,j,k\}$.
   We have $\phi \wedge x_i \vdash_1 \neg x_\ell$.
   Assume, \(\QA{i}\) consists of the two clauses
   \(\neg x_i \vee \neg x_j\) and \(\neg x_i \vee \neg x_k\).
   Then, we have \(\phi\wedge \neg x_j\wedge \neg x_k\vdash_1\neg
      x_\ell\). This implies that \(x_j\vee x_k\vee \neg x_\ell\) is an
   implicate of \(\phi\) which is in contradiction with the
   conditions~\ref{cond:pc-amo-1} and~\ref{cond:pc-amo-2}
used for the formula $\phi \wedge x_\ell$.
\end{proof}

Let \(\varphi(\vek{x}, \vek{y})\) be a \penc{} with \(n\) input
variables \(\vek{x}=(x_1, \ldots, x_n)\).
Recall that Condition~\ref{regular-cond2} of
Definition~\ref{def:regular-form} states that clauses in
\(\QA[\varphi]{i}\) do not contain other input variables than \(x_i\)
for every \(i=1, \ldots, n\), if $\phi(\vek{x}, \vek{y})$ is in regular form.

\begin{proposition}
   \label{prop:one-pure-input-clause}
   Let $\phi(\vek{x}, \vek{y})$ be a prime \penc{} with \(n\geq 4\)
   input variables \(\vek{x}=(x_1, \ldots, x_n)\),
   such that~\ref{regular-cond1} is satisfied and~\ref{regular-cond2}
   is not satisfied. Then there
   is a \penc{} $\phi'$ with \(n-1\) input variables, such that $|\phi| \ge |\phi'| + 3$.
   If $\phi$ is a 2-CNF formula, then so is $\phi'$.
\end{proposition}

\begin{proof}
   Let $i$ be an index, for which~\ref{regular-cond2} is not satisfied
   which means that some clause \(C\in \QA[\varphi]{i}\) contains another
   input variable \(x_j\), \(j\neq i\) in addition to the literal \(\neg
      x_i\).
   By Lemma~\ref{lem:no-pos-and-neg-x} we get that
   $C=\neg x_i \vee \neg x_j$.
   Without loss of generality, assume $i=1$, $j=2$, so $\phi$ contains
   clauses
   $\neg x_1 \vee \neg x_2$,\,
   $\neg x_1 \vee B_1$,\,
   $\neg x_2 \vee B_2$
   for some sets of literals $B_1,B_2$.
   By Lemma~\ref{lem:no-pos-and-neg-x} and
   Lemma~\ref{lem:two-pure-input-clauses}, both $B_1$ and $B_2$ are
   sets of auxiliary literals.
   By Lemma~\ref{lem:amos-fix} we have that $\psi=\phi(\{\neg x_1\})$
   is a \penc{} with \(n-1\) input variables
   \(x_2, \ldots, x_n\). Since \(\varphi\) satisfies~\ref{regular-cond1},
   we have that \(|Q_{\phi, i}|=2\) and
   thus
   \begin{equation}
      \label{eq:pure-1}
      |\phi|\geq|\psi|+2\text{.}
   \end{equation}
   Since $\phi$ satisfies~\ref{regular-cond1},
   the literal \(\neg x_2\) occurs only once in \(\psi\).
   Hence, Proposition~\ref{lem:single-negative} implies that there is a \penc{} $\phi'$
   with \(n-1\) input variables satisfying
   \(|\psi|\geq|\phi'|+1\). Together with~\eqref{eq:pure-1} we get
   \begin{equation*}
      |\phi|\geq|\psi|+2\geq|\phi'|+3
   \end{equation*}
   as required.
\end{proof}

We are now ready to give a proof of
Theorem~\ref{thm:structure-general-CNF-2}.
Since we have already analyzed encodings not satisfying some
of the conditions~\ref{regular-cond1} and~\ref{regular-cond2},
the main step of the proof is to deal with
encodings not satisfying Condition~\ref{regular-cond3}. Let us recall
that~\ref{regular-cond3} requires that for every \(i=1, \ldots, n\)
the set \(\QA[\varphi]{i}\) consists of two binary clauses. Let us
also recall that \(\upclos[\phi]{g_1, \dots, g_k}\) denotes the set of
literals that can be derived by unit resolution from
\(\phi\wedge g_1\wedge\dots\wedge g_k\).

\begin{proof}[Proof of Theorem~\ref{thm:structure-general-CNF-2}]
Let \(\varphi(\vek{x}, \vek{y})\) be a prime \penc{} with input variables
\(\vek{x}=(x_1, \ldots, x_n)\).
If $\phi$ is in regular form, we are done.
If some of the conditions~\ref{regular-cond1} and~\ref{regular-cond2}
is not satisfied, the conclusion follows from
propositions~\ref{prop:number-neg-x-general-CNF}
and~\ref{prop:one-pure-input-clause}. For the rest of the proof
assume that $\phi$ satisfies ~\ref{regular-cond1} and~\ref{regular-cond2}.
If $\phi$ is a 2-CNF, the condition~\ref{regular-cond3} is satisfied
and we are done.

If $\phi$ is not a 2-CNF, assume that $\phi$ does not satisfy~\ref{regular-cond3}
for some \(i\in\{1, \ldots, n\}\).
By Lemma~\ref{lem:amo-pc} we have that one of the clauses in
\(\QA{i}\) is a binary clause. By~\ref{regular-cond1},
we have \(|\QA{i}|=2\) and since \(\QA{i}\) does not
satisfy~\ref{regular-cond3}, the other clause consists of at least
three literals. Moreover, due to~\ref{regular-cond2}
the only input variable which appears in some clause in \(\QA{i}\) is \(x_i\).
Thus we can write
$$
\QA{i}=\{C_1=\neg x_i \vee y, C_2 = \neg x_i \vee z_1 \vee \ldots \vee z_\ell\}
$$
for some literals \(y, z_1, \ldots, z_\ell\) on auxiliary variables where \(\ell>1\).
   We claim that for every \(j\in\{1, \ldots, \ell\}\) we have
   \begin{equation}
      \label{eq:struct-gen-2a}
      \phi\wedge x_i\not\vdash_1\neg z_j
   \end{equation}
   and
   \begin{equation}
      \label{eq:struct-gen-2}
      \phi\wedge y\not\vdash_1\neg z_j\text{.}
   \end{equation}
Let us assume for a contradiction that there is \(j\in\{1, \ldots, \ell\}\)
satisfying the negation of~\eqref{eq:struct-gen-2a} or the negation of~\eqref{eq:struct-gen-2}.
If $\phi\wedge y \vdash_1\neg z_j$, then
$\phi\wedge x_i \vdash_1\neg z_j$, since $C_1 \in \phi$.
Hence, we can assume $\phi\wedge x_i \vdash_1\neg z_j$.
This implies that \(\neg x_i\vee \neg z_j\) is an
implicate of \(\phi\). However, the resolvent
\(\resolvent{\neg x_i\vee \neg z_j}{C_2}\) is a strict subclause of \(C_2\)
which is in contradiction with primality of \(C_2\).

   Consider any input variable \(x_j\), \(j\not= i\).
   Since \(\phi\) satisfies~\ref{cond:pc-amo-2} we have that \(\phi\wedge
      x_i\vdash_1 \neg x_j\). Clause \(C_1\) is the only one in
   \(\phi\) that becomes unit when resolved with \(x_i\).
   Moreover, we can observe that clause \(C_2\) is not used in the
   derivation \(\phi\land x_i\vdash_1\neg x_j\). This follows by~\eqref{eq:struct-gen-2a},
   because in order for \(C_2\) to be used in a unit
   resolution derivation, at least one of \(\neg z_1, \ldots,
      \neg z_\ell\) must be derived first. Thus necessarily
   \begin{equation}
      \label{eq:struct-gen-clos}
      \upclos[\phi]{x_i}=\upclos[\phi]{y}\cup\{x_i\}
   \end{equation}
   and in particular
   \begin{equation}
      \label{eq:struct-gen-1}
      \phi\wedge y\vdash_1\neg x_j\text{.}
   \end{equation}

   Let $\psi = (\phi \setminus \{C_2\}) \cup \{C_3\}$,
   where $C_3 = \neg y \vee z_1 \vee \ldots \vee z_\ell$.
   We prove below that $\psi$ is a \penc{} with input variables
   \(\vek{x}\).
   Since \(|\psi|=|\phi|\) and \(|\psi|\) contains only one
   occurrence of \(\neg x_i\),
   Proposition~\ref{lem:single-negative} implies that there is a \penc{} \(\phi'\)
   with \(n\) input variables satisfying
   \(|\varphi|\geq|\varphi'|+1\) thus
   satisfying condition~\ref{enum:structure-general-CNF-2:1}.
   According to Definition~\ref{def:p-enc} it remains to show that \(\psi\) satisfies
   conditions~\ref{cond:pc-amo-1} and~\ref{cond:pc-amo-2}.

   \begin{amocond}
   \item Let \(x_j\), \(j\in\{1, \ldots, n\}\) be an arbitrary input variable and let us show
      that \(\psi\wedge x_j\) is satisfiable.
      \begin{itemize}
         \item If $j=i$, we have $\phi \wedge x_i \models z_1\vee \cdots\vee z_\ell$,
            since $C_2$ is contained in $\phi$.
            Consequently, $\phi \wedge x_i \models C_3$.
         \item If $j \neq i$, we have $\phi \models \neg y \vee \neg x_j$ by~\eqref{eq:struct-gen-1}. Hence, $\phi \wedge x_j \models \neg y$ and
            $\phi \wedge x_j \models C_3$.
      \end{itemize}
      In both cases, since $\phi \wedge x_j$ is satisfiable, so is
      $\psi \wedge x_j$, and $\psi$ satisfies~\ref{cond:pc-amo-1}.

   \item Let \(j, k\in\{1, \ldots, n\}\) be two different indices of input
      variables and let us
      show that \(\psi\wedge x_j\vdash_1\neg x_k\). Let us look at
      derivation of \(\phi\wedge x_j\vdash_1\neg x_k\).
      \begin{itemize}
         \item If \(j=i\), then clause \(C_2\) is not
            used in the derivation \(\phi\wedge x_i\vdash_1 \neg x_k\).
            This follows by~\eqref{eq:struct-gen-2a},
            because in order for \(C_2\) to be used in a unit
            resolution derivation, at least one of \(\neg z_1, \ldots,
               \neg z_\ell\) must be derived first. It follows that
            \(\psi\wedge x_i\vdash_1\neg x_k\) as well.
         \item Assume \(j\neq i\). If \(C_2\) is
            not used in the derivation \(\phi\wedge
               x_j\vdash_1\neg x_k\), then also \(\psi\wedge
               x_j\vdash_1\neg x_k\) and we are done.
Assume, $C_2$ is used in the derivation $\phi\wedge x_j\vdash_1\neg x_k$.
If \(C_2\) is used to derive some of the literals $z_1, \ldots, z_\ell$,
then in order to do that we need \(\phi\wedge x_j\vdash_1 x_i\), which is not true.
Hence, we can assume that
\(C_2\) is used to derive \(\neg x_i\). Before that we
have \(\phi\wedge x_j\vdash_1 \neg z_r\) for all \(r\in\{1, \ldots, \ell\}\)
and these derivations are possible in \(\psi\) as well. Hence,
\(\psi\wedge x_j\vdash_1\neg z_r\) for all \(r\in\{1, \ldots, \ell\}\).
            Moreover, we obtain \(\psi\wedge x_j\vdash_1 \neg x_i\)
            because we can replace the step using \(C_2\) in the original unit
            resolution derivation with two steps using \(C_3\) to
            derive \(\neg y\) and then \(C_1\) to derive \(\neg x_1\).
            Together, we get that also in this case
            \(\psi\wedge x_j\vdash_1\neg x_k\).
      \end{itemize}
   \end{amocond}
   This concludes the proof of
   Theorem~\ref{thm:structure-general-CNF-2}.
\end{proof}

The following notation will be used in the subsequent sections.
Assume, \(\phi(\vek{x}, \vek{y})\) is a \penc{} with input variables
\(\vek{x}=(x_1, \ldots, x_n)\) in regular form.
For each \(i=1,\dots,n\), let
\begin{align}
   \label{eq:PA}
   \PA[\phi]{i}&=\left\{ e   \;|\; (\neg x_i\vee e)\in\QA{i}\right\}\text{,}\\
   \label{eq:PB}
   \PB{i}&=\var{\PA[\phi]{i}}\text{.}
\end{align}
Clearly, for all \(i=1, \ldots, n\), we have \(|\PA[\phi]{i}|=2\),
Moreover, since $\neg x_i$ is not an implicate of $\phi$, the
set $\PA[\phi]{i}$ does not contain complementary literals
and we have $|\PB{i}|=2$. The following is now an easy observation.

\begin{lemma} \label{lem:PA-different}
If \(i, j\in\{1, \ldots, n\}\) are two different indices of input
variables, then \(\PA{i}\neq\PA{j}\).
\end{lemma}

\begin{proof}
Assuming \(\PA{i}=\PA{j}\), we get a contradiction with
conditions~\ref{cond:pc-amo-1} and~\ref{cond:pc-amo-2} as follows.
The formula \(\phi\wedge x_i\) is satisfiable and
derives both literals in $\PA{i}=\PA{j}$. Hence, it is
not possible to have \(\phi\wedge x_i\vdash_1 \neg x_j\).
\end{proof}

Note that it follows from Lemma~\ref{lem:PA-different} that
\(\varphi\) contains at least \(\sqrt{n/2}\) auxiliary variables in
\(\vek{y}\). This is because if the number of auxiliary variables is
\(\ell\), then we have at most \(\binom{2\ell}{2}\) sets
\(\PA{i}\). It follows that \(\binom{2\ell}{2}\geq n\) and
\(\ell\geq\sqrt{n/2}\) and this gives a hint on how \(\sqrt{n}\) gets into
the lower bound. In order to obtain a corresponding term
$\Omega(\sqrt{n})$ in the lower bound on the number of clauses,
we distinguish two types of clauses in a \penc{} in regular form as follows:
\begin{itemize}
   \item Clauses from \(\bigcup_{i=1}^n\QA{i}\) are of \emph{type \Qtype{}}.
   \item
      The remaining clauses in \(\phi\) are of \emph{type \Rtype{}}.
\end{itemize}
A \penc{} in regular form with $n$ input variables
contains $2n$ clauses of type \Qtype{}. This was used in
Corollary~\ref{cor:two-n} to prove a lower bound $2n$
on the size of every \penc{} with a sufficiently large number \(n\)
of input variables. In order to prove the lower bounds stated
in the introduction, we prove in the next two sections
lower bounds on the number of clauses of type \Rtype{} in cases of
general CNF \pencs{} and 2-CNF \pencs{}.

   \section{A Lower Bound For General P-Encodings}
\label{sec:gen-lb} % chktex 24

This section is devoted to the proof of a lower bound on
\(\pencS{n}\), i.e.\ the proof of statement~\ref{enum:lower-bound:1} for
\(\pencS{n}\) and statement~\ref{enum:lower-bound:2} of
Theorem~\ref{thm:lower-bound}. The proof consists of a lower bound on
the size of a \penc{} in regular form and an inductive argument
based on Theorem~\ref{thm:structure-general-CNF-2}. We have observed
at the end of Section~\ref{sec:regular-form} that if
\(\varphi(\vek{x}, \vek{y})\) is a \penc{} with \(n\) input variables
in regular form, then the number of auxiliary variables is at least
\(\sqrt{n/2}\). We improve this bound and, moreover, we
show that there must be an input variable \(x_i\) such that unit
propagation starting from \(\varphi\land x_i\) derives
at least \(\sqrt{n-3/4}+1/2\) literals on auxiliary variables. This implies
that there must be almost as many \Rtype{} clauses in \(\varphi\)
implying the lower bound.

We start with the base cases for induction. By CNF complexity of
a boolean function, we mean the minimum size of a CNF formula
expressing the function. Clearly, this is also the minimum size of an
encoding of the function without auxiliary variables. One can easily
verify that the CNF complexity of $\amo_n$ is $\binom{n}{2}$ and the
CNF complexity of $\exone_n$ is $\binom{n}{2} + 1$.

\begin{lemma}
   \label{lem:base-cases}
Let \(\amoS{n}\) denote the minimum size of a PC encoding of \(\amo_n\)
and let \(\exoneS{n}\) denote the minimum size of a PC encoding of
\(\exone_n\). Then, we have
   \begin{alignat}{3}
      \amoS{2}&=1,\quad&\exoneS{2}&=2,\quad&\pencS{2}&=1\\
      \amoS{3}&=3,\quad&\exoneS{3}&=4,\quad&\pencS{3}&=3\text{.}
   \end{alignat}
\end{lemma}

\begin{proof}
Representations of $\amo_n$ and $\exone_n$ containing all the
prime implicates achieve these bounds, are the smallest possible,
and are propagation complete.
This implies the required upper bounds on $\amoS{n}$, $\exoneS{n}$,
and $\pencS{n}$ for $n=2,3$.

Lemma~\ref{lem:aux-var-occurrences} implies that if a function
has a representation of size at most $4$, then the
size of the smallest encoding (not necessarily propagation complete)
is equal to the size of the smallest representation.
This implies the stated values of $\amoS{n}$ and $\exoneS{n}$.

As explained in Section~\ref{sec:relations-p-enc}, every \penc{} with \(n\)
input variables is either an encoding of $\amo_n$ or an encoding of $\exone_n$.
Hence, the lower bounds from the previous paragraph hold also for \pencs{}
and this implies the stated values of $\pencS{n}$.
\end{proof}

Let us prove additional properties of sets \(\PA{i}\)
introduced at the end of Section~\ref{sec:regular-form}.

\begin{lemma}
   \label{lem:star-leaves}
   Let \(\phi\) be a \penc{} with input variables \((x_1, \ldots, x_n)\) in regular form.
   Let $i,j,k$ be
   different indices with $\PA{i} = \{g, h_1\}$, $\PA{j} = \{g,
      h_2\}$, and $\PA{k} = \{g, h_3\}$ for
   $g,h_1,h_2,h_3\in\lit{\vek{y}}$.
   Then variables \(\var{h_1}\), \(\var{h_2}\), and \(\var{h_3}\)
   are pairwise different.
\end{lemma}

\begin{proof}
   Let us show by contradiction that \(\var{h_1}\neq\var{h_2}\).
   To this end, assume \(\var{h_1}=\var{h_2}\). Since \(\PA{i}\neq\PA{j}\)
by Lemma~\ref{lem:PA-different}, we have
   \(h_1 = \neg h_2\).
   By condition~\ref{cond:pc-amo-2}, \(\phi\wedge x_k\vdash_1\neg x_i\)
   and $\phi\wedge x_k\vdash_1\neg x_j$. Since \(\phi\wedge
      x_k\vdash_1 g\), necessarily
   \(\phi\wedge x_k\vdash_1\neg h_1\) and \(\phi\wedge
      x_k\vdash_1\neg h_2\). However, then \(\phi\wedge
      x_k\vdash_1\bot\) which is in contradiction with~\ref{cond:pc-amo-1}.

   The remaining cases, i.e. \(\var{h_1}\neq\var{h_3}\) and
   \(\var{h_2}\neq\var{h_3}\) are symmetrical.
\end{proof}

Using Lemma~\ref{lem:star-leaves} for all triples of indices
of input variables, we obtain the following.

\begin{corollary}
\label{cor:star-leaves-gen}
Assume, \(\phi(\vek{x}, \vek{y})\) is a \penc{} with input variables
\(\vek{x}=(x_1, \ldots, x_n)\) in regular form. Let
\[
   I_h=\{i\in\{1, \ldots, n\}\mid h\in\PA{i}\}\text{,}
\]
where $h \in \lit{\vek{y}}$, and define
\[
   L_h=\bigcup_{i\in I_h}\PA{i} \ .
\]
If $|I_h|\geq 3$, then $L_h$ contains literals on different
variables and \(|L_h|=|I_h|+1\).
\end{corollary}
\begin{proof}
   This is a corollary of Lemma~\ref{lem:star-leaves}. If we
   remove literal \(h\) from each \(\PA{i}\), \(i\in I_h\), then
   the remaining literals are on pairwise different variables
   different from $\var{h}$.
\end{proof}

We are now ready to show the lower bound on the size of a \penc{} in
regular form.

\begin{lemma}
   \label{lem:regular-size}
   If $\phi(\vek{x},\vek{y})$ is a \penc{} with \(n\geq
      7\) input
   variables \(\vek{x}=(x_1, \ldots, x_n)\)
in regular form, then $|\phi| \ge 2n + \sqrt{n-3/4} - 3/2$.
\end{lemma}

\begin{proof}
Since the formula $\phi(\vek{x},\vek{y})$ contains $2n$ clauses of type \Qtype{},
it is sufficient to prove that it contains at least \(\sqrt{n-3/4}-3/2\)
clauses of type \Rtype{}.

   Let $L=\bigcup_{i=1}^n \PA[\phi]{i}$ be the set of
   auxiliary literals in clauses of type \Qtype{}.
   For each \(h\in L\), let \(I_h\) and \(L_h\) be defined as in
   Corollary~\ref{cor:star-leaves-gen}.
   Choose $g\in L$ that maximizes \(\left|I_g\right|\),
fix some $i \in I_g$ and denote
$M_i=\upclos[\phi]{x_i} \cap \lit{\vek{y}}$.
Each literal in $M_i \setminus \PA[\phi]{i}$ is derived by
unit propagation from $\phi \wedge x_i$
using a different clause of type \Rtype{}. Hence, the number of
the clauses of type \Rtype{} is at least $|M_i|-2$.

For each \(j\in\{1, \ldots, n\}\setminus\{i\}\), the literal $\neg x_j$
must be derived from some \(h\in M_i\) using a clause
of type \Qtype{} containing $\neg h$.
As $\left|I_{\neg h}\right|\le\left|I_g\right|$ for every \(\neg h \in L\),
each \(h\in M_i\) can derive $\neg x_j$ for at most $\left|I_g\right|$
   values of $j$. Thus,
   \begin{equation}
      \label{eq:lower-bound-UP-dir2}
      \left|M_i\right|\cdot\left|I_g\right| \ge n-1\text{.}
   \end{equation}

If $\left|I_g\right| \ge 3$,
then according to Corollary~\ref{cor:star-leaves-gen}, we get \(|L_g|=|I_g|+1\).
In order to derive all $\neg x_j$, \(j\in I_g\setminus\{i\}\) from \(\phi \wedge x_i\),
$M_i$ has to contain the negations of the literals in \(L_g\setminus\PA{i}\).
Since $M_i$ contains also $\PA{i}$ and these literals are on
variables not in $\var{L_g\setminus\PA{i}}$, we have
   \begin{equation}
      \label{eq:lower-bound-UP-dir1}
      \left|M_i\right| \ge \left|I_g\right|+1\text{.}
   \end{equation}

Combining the two bounds, we get
$$
\left|M_i\right| \ge \max\{|I_g|+1, (n-1)/|I_g|\}
$$
as follows. If $\left|I_g\right| \ge 3$, the
claims~\eqref{eq:lower-bound-UP-dir1} and~\eqref{eq:lower-bound-UP-dir2} apply.
If $\left|I_g\right| \le 2$, observe that by~\eqref{eq:lower-bound-UP-dir2},
$\left|M_i\right| \ge (n-1)/2 \ge 3 \ge |I_g|+1$, since $n \ge 7$.

For \(s \ge 1\), the smallest value of the function $s\mapsto\max\{s+1, (n-1)/s\}$ 
is $s_0+1$ for $s_0 \ge 1$, such that $s_0 + 1 = (n-1)/s_0$.
Hence, we have $|M_i| \ge \sqrt{n-3/4}+1/2$ and the number of clauses
of type \Rtype{} is at least $\sqrt{n-3/4}-3/2$.
\end{proof}

Let us conclude the section with the following theorem which
represents the statement~\ref{enum:lower-bound:1} for
\(\pencS{n}\) and statement~\ref{enum:lower-bound:2} of
Theorem~\ref{thm:lower-bound}.

\begin{theorem}\label{thm:general-lb-3}
   For \(n\ge 3\), the minimum size $\pencS{n}$ of a \penc{} with
   \(n\) input variables satisfies
   \begin{enumerate}
      \item If $n \le 8$, then $\pencS{n} = 3n-6$.
      \item
         If $n \ge 9$, then
         $\pencS{n} \ge 2n + \sqrt{n} - 2$.
   \end{enumerate}
\end{theorem}

\begin{proof}We treat the two claims separately:
   \begin{enumerate}
      \item It was shown in Lemma~\ref{lem:simple-upper-bound}
and Lemma~\ref{lem:p-enc} that
         $\pencS{n} \le 3n-6$. To show that
$\pencS{n} \ge 3n-6$ for $3 \le n \le 8$,
we proceed by induction on $n$. The basis
         $\pencS{3} = 3$ is given by Lemma~\ref{lem:base-cases}.
         Let \(\phi\) be a prime \penc{} with \(n \ge 4\) input
         variables of size $\pencS{n}$.
Theorem~\ref{thm:structure-general-CNF-2}
implies that either \(\left|\phi\right|\ge\pencS{n-1}+3\) or \(\phi\) is
in regular form. In the first case, the induction hypothesis implies
$$
|\phi| \ge \pencS{n-1} + 3 \ge 3n-6\text{.}
$$
If $\phi$ is in regular form and $n \le 6$, then the structure of
the regular form implies $|\phi| \ge 2n \ge 3n-6$.
If $\phi$ is in regular form and $7 \le n \le 8$,
then Lemma~\ref{lem:regular-size} implies
$$
|\phi| \ge \left\lceil 2n + \sqrt{n - 3/4} - 3/2 \right\rceil \ge 3n-6 \ .
$$
\item We proceed by induction on $n$ using the previous case as the basis.
   Let \(\phi\) be a prime \penc{} with \(n \ge 9\)
   input variables of size $\pencS{n}$.
It follows from Theorem~\ref{thm:structure-general-CNF-2} that
either \(\left|\phi\right|\ge\pencS{n-1}+3\) or \(\phi\) is in regular form.
In the first case, observe that:
\begin{itemize}
\item If \(n=9\), we obtain \(\pencS{n-1}+3 = 21 > 2n + \sqrt{n} - 2\) by the
      first claim of this theorem.
\item If \(n \ge 10\), the induction hypothesis and $\sqrt{n-1} + 1 > \sqrt{n}$ imply
$$
\pencS{n-1}+3 \ge 2(n-1) + \sqrt{n-1} + 1 \ge 2n + \sqrt{n} - 2
$$
as required.
\end{itemize}
If $\phi$ is in regular form, Lemma~\ref{lem:regular-size} implies
$$
\left|\phi\right| \ge 2n + \sqrt{n - 3/4} - 3/2 \ge 2n + \sqrt{n} - 2
$$
as required.
   \end{enumerate}

\end{proof}

Note that the upper bound \(3n-6\) in Theorem~\ref{thm:general-lb-3} for
\(n\leq 8\) is achieved by the sequential encoding of \(\amo_n\)
described in Section~\ref{ssec:sequential-enc}. It follows that
actually \(\pencS{n}=\amoS{n}=3n-6\) in this case where \(\amoS{n}\)
denotes the minimum size of a PC encoding of \(\amo_n\).

   \section{A Lower Bound For 2-CNF P-Encodings}
\label{sec:quad-lb} % chktex 24

In this section we prove the remaining parts of
Theorem~\ref{thm:lower-bound}.
In particular, we prove a lower bound \(2n+2\sqrt{n}-3\)
on $\pencQS{n}$ and prove that $\pencQS{n}$ is equal to the
minimum size of a 2-CNF encoding of $\amo_n$, even if we
do not require propagation completeness.
Importance of the special case of 2-CNF encodings for $\amo_n$ comes from the fact
that the smallest known encodings are in 2-CNF as well as all the other
encodings suggested in the literature.

We have already argued at
the beginning of Section~\ref{sub:main-result} that a 2-CNF \penc{}
\(\varphi\)
with \(n\geq 3\) input variables cannot encode \(\exone_n\). It
follows by Lemma~\ref{lem:p-enc-is-enc} that \(\varphi\) is an
encoding of \(\amo_n\).
The correspondence holds also in the opposite direction
in the following sense.

\begin{lemma} \label{lem:minimum-2-CNF-AMO}
A minimum size 2-CNF encoding of $\amo_n$ does not have unit
implicates, is a prime \penc{}, and its size is equal to $\pencQS{n}$.
\end{lemma}

\begin{proof}
Let $\phi(\vek{x}, \vek{y})$ be a minimum size 2-CNF encoding of $\amo_n$.
Since $\amo_n$ does not have unit implicates, $\phi$ does not have
a unit implicate on an input variable. Assume for a contradiction
that $\phi$ has a unit implicate $l$ on an auxiliary variable.
Setting this variable, so that $l$ is satisfied, leads to
a PC encoding of $\amo_n$. Moreover, since the literal $l$ has at least
one occurence in $\phi$, at least one clause can be removed and
this contradicts the assumption. Hence, $\phi$ does not have
unit implicates. It follows that \(\varphi\) is a
prime 2-CNF and by results of~\cite{BBCGK13} we get that \(\varphi\)
is propagation complete.
   It follows that it is a \penc{} and thus
   \(|\varphi|\geq\pencQS{n}\). On the other hand, we have observed above
   that a 2-CNF \penc{} with \(n\geq 3\) inputs is an encoding of
   \(\amo_n\) and thus \(\pencQS{n} \ge |\varphi|\).
Hence, we have $|\phi| = \pencQS{n}$ for all $n \ge 1$, since the
cases $n \le 2$ can easily be verified.
\end{proof}

It follows that refering to minimum size 2-CNF encodings of
\(\amo_n\) is the same as refering to minimum size 2-CNF \pencs{}.
In order to prove a lower bound on the size of a 2-CNF encoding of \(\amo_n\),
we use Theorem~\ref{thm:structure-general-CNF-2} similarly as in
Section~\ref{sec:gen-lb} to handle encodings, which are not in regular form.
However, the analysis of encodings in regular form is different
and implies a stronger lower bound.

One of the differences in case of 2-CNF formulas is that a unit
resolution derivation can be reversed in the following sense.
If one step of unit propagation can derive $h$ from $\neg g$, then we
can also derive $g$ from $\neg h$ in one step, since both these
steps are possible if and only if the clause $g \vee h$ is contained in the
formula. Due to this, it is useful to
represent a 2-CNF formula with an \emph{implication graph} introduced
in~\cite{APT79} as follows.
For a formula $\phi = \phi(\vek{z})$,
let $G_\phi=(V,E)$ be a directed graph with the set of vertices $V = \lit{\vek{z}}$
and with the set of directed edges $E$ containing \((\neg g, h)\) and \((\neg h, g)\)
for each clause \((g\vee h) \in \phi\). These two directed edges (arcs) represent the
implications \(\neg g\to h\) and \(\neg h\to g\), respectively.
This graph is \emph{skew-symmetric} (also called duality property
in~\cite{DP60} and mirror property in~\cite{CH11}), meaning that $(g,h) \in E$ if and only
if $(\neg h, \neg g) \in E$. We can exploit the properties of the
implication graph to show
stronger properties of the sets of literals \(\PA[\phi]{i}\) defined
in~\eqref{eq:PA} than in the case of general CNF encodings.
In particular, we show that the part of the analysis
that determines the constant of the term $\sqrt{n}$ in the lower bound
can be reduced to the case where the sets
\(\PB[\phi]{i}=\var{\PA{i}}\) are pairwise different
for \(i=1, \ldots, n\). Recall that in the general case we were only able
to show in Lemma~\ref{lem:PA-different} that the sets \(\PA[\phi]{i}\)
are pairwise different.

The following two basic properties of 2-CNF formulas are used implicitly throughout this section.
See, e.g., Theorem 5.6 in~\cite{CH11} for the omitted proof of
Lemma~\ref{lem:Chains-2CNF-general}.

\begin{lemma} \label{lem:Chains-2CNF-general}
   Let $\phi(\vek{z})$ be a 2-CNF formula. Then, for each $g \in \lit{\vek{z}}$, we have
   $\phi \vdash_1 g$
   if and only if there is a literal \(h\) which forms a unit clause
   in \(\phi\) and there is a path from \(h\) to \(g\) in \(G_\phi\).
\end{lemma}

We are mainly interested in the properties of the minimum size 2-CNF
encodings of $\amo_n$ and these encodings do not contain unit clauses
by Lemma~\ref{lem:minimum-2-CNF-AMO}.

\begin{lemma}
   \label{lem:Chains-2cnf-sym}
Let \(\phi(\vek{z})\) be a 2-CNF formula not containing unit
clauses and let \(g, h\in\lit{\vek{z}}\).
Then the following conditions are equivalent:
   \begin{enumerate}[label={(\roman*)}]
      \item\label{enum:Chains-2cnf-sym:ij}\(\phi\wedge g\vdash_1\neg h\),
      \item\label{enum:Chains-2cnf-sym:ji}\(\phi\wedge h\vdash_1\neg g\).
   \end{enumerate}
\end{lemma}
\begin{proof}
Assume~\ref{enum:Chains-2cnf-sym:ij}.
By Lemma~\ref{lem:Chains-2CNF-general}, there is a path in $G_\phi$
from a unit clause in the formula $\phi \wedge g$ to $\neg h$.
By the assumption, this path starts in $g$.
Since $G_\phi$ is skew-symmetric, it contains also
a path from $h$ to $\neg g$ and, hence, we have~\ref{enum:Chains-2cnf-sym:ji}.
By symmetry, also~\ref{enum:Chains-2cnf-sym:ji}
implies~\ref{enum:Chains-2cnf-sym:ij}.
\end{proof}

The following series of statements analyze properties of
minimum size 2-CNF encodings of \(\amo_n\) and finally
leads to Theorem~\ref{thm:Main}.

\begin{lemma}
   \label{lem:Chains}
   Let \(\phi(\vek{x}, \vek{y})\) be a minimum size 2-CNF encoding of
   \(\amo_n\) and let \(i,j\in\{1,\dots,n\}\), \(i\neq j\). Then there exist literals
   \(e_0, e_1, \dots, e_p\in \lit{\vek{x}\cup\vek{y}}\), \(p\geq 1\), such that:
   \begin{enumerate}[label={(\roman*)}]
      \item\label{enum:Chains:1}
         \(e_0=x_i\), \(e_p=\neg x_j\), and \(e_0, e_1, \dots, e_p\) form a path in \(G_\phi\),
      \item\label{enum:Chains:3}
         \(e_1, \dots, e_{p-1}\) are
         pairwise distinct literals from \(\lit{\vek{y}}\).
   \end{enumerate}
\end{lemma}
\begin{proof}
By assumption, $\phi \not\vdash_1 \neg x_j$ and \(\phi\wedge x_i\vdash_1\neg x_j\).
Hence, by Lemma~\ref{lem:Chains-2CNF-general},
there is a path from $x_i$ to $\neg x_j$ in $G_\phi$.
   Let  \(x_i=e_0, e_1, \dots, e_p=\neg x_j\) be a shortest such path.
If $p=1$, the proof is finished.
   Suppose that \(p\geq 2\)
   and let us show that the path meets~\ref{enum:Chains:3}. Since the sequence is the shortest
   possible and \(\phi\wedge x_i\not\vdash_1\bot\), the literals \(e_0, \dots, e_p\) are on pairwise
   different variables. Assume for a contradiction that there is a literal \(e_q\)
   on an input variable \(x_k\), \(k\not\in\{i, j\}\). If
   \(e_q=x_k\), then \(\phi\wedge x_i\vdash_1
      x_k\) and if \(e_q=\neg x_k\), then \(\phi\wedge
      \neg x_k\vdash_1 \neg x_j\).
Both these cases contradict Lemma~\ref{lem:amo-pc}.
\end{proof}

The following proposition shows that we can restrict our consideration
to 2-CNF encodings of \(\amo_n\) with no positive occurrences of
input variables. This is a difference from the case of
general PC encodings of $\amo_n$. In Section~\ref{sec:further}, we present an example
of an irredundant prime PC encoding of $\amo_n$ of size $\Theta(n^2)$ containing positive
occurrences of input variables. We believe that positive occurrences
of input variables cannot occur in a minimum size prime PC encoding of $\amo_n$, however,
we do not have a proof of this conjecture.

\begin{lemma} \label{lem:PositiveInpVars}
A minimum size 2-CNF encoding of \(\amo_n\) contains no positive
occurrence of an input variable.
\end{lemma}

\begin{proof}
Assume, $\phi$ is a 2-CNF encoding of $\amo_n$ of minimum size
and contains a clause containing a positive literal on an input variable.
Let $\phi'$ be the formula obtained by removing all such clauses
from $\phi$. Since $\phi'$ is a subset of $\phi$, it satisfies~\ref{cond:pc-amo-1}.
Moreover, by Lemma~\ref{lem:Chains}, for every $i,j \in \{1,\ldots,n\}$,
there is a series of resolutions witnessing $\phi \wedge x_i \vdash_1 \neg x_j$
that does not derive any literal on an input variable except of $\neg x_j$.
The clauses used in this series of resolutions contain only the literals
$\neg x_i$, $\neg x_j$, and literals on auxiliary variables.
Hence, this series of resolutions can be used also in $\phi'$ and
this implies that $\phi'$ satisfies~\ref{cond:pc-amo-2}. This is
a contradiction with minimality of $\phi$.
\end{proof}

Using Lemma~\ref{lem:minimum-2-CNF-AMO} and
Theorem~\ref{thm:structure-general-CNF-2}, we obtain that
a minimum size 2-CNF encoding \(\phi\) of \(\amo_n\) either
satisfies \(|\phi|\geq\pencQS{n-1}+3\) or it is in regular form
introduced in Definition~\ref{def:regular-form}.
The purpose of several following propositions
concluded by Proposition~\ref{lem:quad-regular}
is to show that in case of
2-CNF encodings we can reduce the analysis to an even more restricted form.
Recall the notation $\PA[\phi]{i}$ from~\eqref{eq:PA}.
In addition to the conditions of the regular form
this restricted regular form requires that the following two conditions
are satisfied.
\begin{itemize}
   \item The sets \(\PB{1}, \ldots, \PB{n}\) are pairwise distinct.
   \item For any three different indices
      \(r, s, t\in\{1, \ldots, n\}\) we have
      \(\left|\PB{r}\cup\PB{s}\cup\PB{t}\right|> 3\).
\end{itemize}
Lemma~\ref{lem:Equal-PB} and Proposition~\ref{lem:Equal-PB-Cont} show how to deal
with encodings not satisfying the first of these conditions.

\begin{lemma}
   \label{lem:Equal-PB}
Let \(\phi(\vek{x},\vek{y})\) be a minimum size 2-CNF encoding
of \(\amo_n\) in regular form, $n \ge 3$. Let \(r, s\in\{1, \ldots, n\}\) be
   different and let us suppose \(\PB{r}=\PB{s}\) and
   \(\PA{r}=\{g,h\}\) for \(g,h\in\lit{\vek{y}}\). Then
   \(\PA{s}=\{\neg g,\neg h\}\).
\end{lemma}

\begin{proof}
Assume without loss of generality \(r=1\) and \(s=2\).
Hence, we have $\PA{1}=\{g,h\}$.
By Lemma~\ref{lem:PA-different}, $\PA{2} \not= \{g, h\}$.
In order to prove the lemma, it is sufficient to exclude
the cases $\PA{2}=\{\neg g, h\}$ and $\PA{2}=\{g, \neg h\}$.
Since these cases are symmetrical, we consider only the
first of them.

Assume \(\PA{2}=\{\neg g,h\}\) and let
         \begin{equation}
            \label{eq:PB-mod-1}
            \phi'=(\phi\cup \{\neg x_1 \vee \neg x_2\} )\setminus\{\neg x_1 \vee g,\, \neg x_2 \vee \neg g\}\text{.}
         \end{equation}
         We show that \(\phi'\) satisfies the
         conditions~\ref{cond:pc-amo-1} and~\ref{cond:pc-amo-2},
         thus it encodes \(\amo_n\),
         which contradicts minimality of \(\phi\).
\begin{amocond}
\item For every $i$, the formula $\phi \wedge x_i$ is satisfiable.
Since the added clause \(\neg x_1 \vee \neg x_2\) is an
implicate of \(\amo_n\) and removing clauses preserves any
satisfying assignments, the formula $\phi' \wedge x_i$
is satisfiable as well.
\item Let us prove $\phi'\wedge x_i \vdash_1 \neg x_j$ for every
\(i,j\in \{1,\dots,n\}\), \(i\neq j\).
Due to Lemma~\ref{lem:Chains-2cnf-sym}, we can assume \(i<j\).
If \(i,j\not\in\{1,2\}\), then
\(\phi'\wedge x_i \vdash_1 \neg x_j\) using the chain \(e_0,\dots,e_p\)
provided by Lemma~\ref{lem:Chains} applied to
\(\phi\). This chain is not affected by~\eqref{eq:PB-mod-1} and
can thus be used in \(\phi'\) as well. Let \(i\in\{1,2\}\).
The case \(i=1,\, j=2\) is trivial, so let \(j>2\).
The set $\upclos[\phi]{x_j}$ contains a negation of a literal
from each of the sets $\PA{1}$ and $\PA{2}$. Since it cannot contain
both $g$ and $\neg g$, it contains $\neg h$ and by Lemma~\ref{lem:Chains},
there is a path in $G_\phi$ from $x_j$ to $\neg h$ that does not contain a literal
on an input variable except of the starting node.
By Lemma~\ref{lem:Chains-2cnf-sym}, this implies that
\(\phi\wedge h \vdash_1 \neg x_j\)
by a path that does not use the clauses omitted in $\phi'$,
so we have \(\phi'\wedge x_i \vdash_1 \neg x_j\).
\end{amocond}
\end{proof}

The size of an encoding not satisfying the first condition
of restricted regular form can be estimated as follows.

\begin{proposition}
   \label{lem:Equal-PB-Cont}
   Let \(n\geq 5\) and
   let \(\phi(\vek{x},\vek{y})\) be a minimum size 2-CNF
   encoding of \(\amo_n\) with the minimum number of auxiliary variables
and, moreover, assume that $\phi$ is in regular form.
If there are two different indices \(r, s\in\{1, \ldots, n\}\),
such that \(\PA{r}=\{g,h\}\) and
   \(\PA{s}=\{\neg g,\neg h\}\) for \(g,h\in \lit{\vek{y}}\), then
   \(\left|\phi\right|\ge\pencQS{n-2}+7\).
\end{proposition}
\begin{proof}
Without loss of generality, assume \(r=1\) and \(s=2\).
Hence, $\PA{1}=\{g,h\}$ and $\PA{2}=\{\neg g,\neg h\}$.
Denote
   \begin{align*}
      A&=\{j\in\{3,\dots,n\}\mid \text{\(\phi\wedge x_j\vdash_1 \neg
               g\) and \(\phi\wedge x_j\vdash_1 h\)}\}\text{,} \\
      B&=\{j\in\{3,\dots,n\}\mid \text{\(\phi\wedge x_j\vdash_1 g\)
            and \(\phi\wedge x_j\vdash_1 \neg h\)}\}\text{.}
   \end{align*}
For a proof of \(A\cup B = \{3,\dots,n\}\), assume $j \ge 3$.
The set $\upclos[\phi]{x_j}$ contains a negation of a literal
from each of the sets \(\PA{1}=\{g,h\}\) and \(\PA{2}=\{\neg g,\neg h\}\),
however, it cannot contain complementary literals.
It follows that either $\upclos[\phi]{x_j} = \{\neg g,h\}$
or $\upclos[\phi]{x_j} = \{g,\neg h\}$, so $j \in A\cup B$.

Let us prove that \(A\neq\emptyset\) and \(B\neq\emptyset\)
by contradiction.
If \(A=\emptyset\), we prove that the formula
\(\phi' = \phi\setminus \{\neg x_1 \vee g\}\) encodes \(\amo_n(\vek{x})\)
by verifying the conditions~\ref{cond:pc-amo-1} and~\ref{cond:pc-amo-2}
in contradiction to minimality of \(\phi\):
   \begin{amocond}
   \item For every $i\in\{1,\ldots,n\}$, $\phi' \wedge x_i$ is satisfiable, since
   $\phi \wedge x_i$ is satisfiable and $\phi'$ is a subset of $\phi$.
   \item  Let us verify \(\phi\wedge x_i\vdash_1 \neg x_j\)
      for each \(i,j\in\{1,2\} \cup B\), \(i\neq j\).
      Due to Lemma~\ref{lem:Chains-2cnf-sym}, we can assume \(i<j\).
      \begin{itemize}
         \item If \(i=1\), then \(j\in B\cup\{2\}\). Thus,  \(\phi\wedge x_1\vdash_1 h\)
            and  \(\phi\wedge x_j\vdash_1 \neg h\). By
            Lemma~\ref{lem:Chains-2cnf-sym} we get
$\phi\wedge h \vdash_1\neg x_j$ and, hence,
\(\phi\wedge x_1\vdash_1\neg x_j\).
         \item If \(i\ge 2\), then Lemma~\ref{lem:Chains} guarantees that \(\phi\wedge x_i\vdash_1 \neg x_j\) is witnessed by a series of unit resolutions not using the clause \(\neg x_1 \vee g\).
      \end{itemize}
   \end{amocond}
   Similarly, if  \(B=\emptyset\), then
   \(\phi\setminus \{\neg x_1 \vee h\}\) encodes \(\amo_n(\vek{x})\)
in contradiction to minimality of $\phi$.
Altogether, \(A\neq\emptyset\) and \(B\neq\emptyset\).

For the last step of the proof, consider the formula
\[\phi'= \phi\setminus(\gamma \cup \{(\neg x_1 \vee h),(\neg x_2 \vee \neg h) \})\text{, }\]
where \(\gamma=\{C\in\phi\mid g\in C \text{ or } \neg g \in C\}\).
By Lemma~\ref{lem:aux-var-occurrences} we get that \(|\gamma|\geq 5\)
and \(\left|\phi\right|\ge|\phi'|+7\).
Let us show that \(\phi'\) encodes \(\amo_{n-2}(x_3,\dots,x_n)\)
by verifying the conditions~\ref{cond:pc-amo-1} and~\ref{cond:pc-amo-2}.
\begin{amocond}
\item For every $i\in\{3,\ldots,n\}$, $\phi' \wedge x_i$ is satisfiable, since
$\phi \wedge x_i$ is satisfiable and $\phi'$ is a subset of $\phi$.
\item Let us verify \(\phi'\wedge x_i\vdash_1 \neg x_j\)
for each \(i,j\in\{3,\dots,n\}\).
Each of the sets $\{g, h\}$, $\{g, \neg h\}$, $\{\neg g, h\}$, $\{\neg g, \neg h\}$
is a subset of $\upclos[\phi]{x_k}$ for some $k\in\{1, \ldots, n\}$.
Hence, each of the formulas
$$
\phi\wedge \neg g \wedge h,\;
\phi\wedge g \wedge \neg h,\;
\phi\wedge g \wedge h,\;
\phi\wedge \neg g \wedge \neg h
$$
is satisfiable. Distinguish the following cases:
      \begin{itemize}
         \item If \(i\in A\) and \(j \in B\), let
            \(e_0,\dots,e_p\in\lit{\vek{x}\cup \vek{y}}\) be a chain of literals
            derived in a series of unit resolutions witnessing
            \(\phi\wedge x_i\vdash_1 h \). If \(e_q\in\{g,\neg g\}\) for some
            \(q\in\{1,\dots,p-1\}\), then \(\phi\wedge e_q\vdash_1 h \),
            which contradicts \(\phi\wedge e_q\wedge \neg h\) being satisfiable.
            Thus, the chain is present in \(\phi'\) as well and we have
            \(\phi'\wedge x_i \vdash_1 h\).
On the other hand, we have $\phi \wedge x_j \vdash_1 \neg h$ and
by a similar argument, we obtain \(\phi'\wedge x_j \vdash_1 \neg h\).
This implies $\phi'\wedge h \vdash_1 \neg x_j$ and thus
            \(\phi'\wedge x_i \vdash_1 \neg x_j\).

         \item The case of \(i\in B\), \(j \in A\) follows from the previous one
by Lemma~\ref{lem:Chains-2cnf-sym}.
         \item If \(i,j\in A\), let
            \(e_0,\dots,e_p\in\lit{\vek{x}\cup \vek{y}}\) be a chain of literals
            derived in a series of unit resolutions witnessing
            \(\phi\wedge x_i\vdash_1 \neg x_j\) according to Lemma~\ref{lem:Chains}
that does not use the clauses $\neg x_1 \vee h$ and $\neg x_2 \vee \neg h$.
Assume, for a contradiction, \(e_q\in\{g,\neg g\}\) for some \(q\in\{1,\dots,p-1\}\).
Then, we have either $\phi \wedge x_i \vdash_1 g$ or $\phi \wedge x_j \vdash_1 g$
by Lemma~\ref{lem:Chains-2cnf-sym}.
This is not possible, since $i,j \in A$ implies
$\phi \wedge x_i \vdash_1 \neg g$ and $\phi \wedge x_j \vdash_1 \neg g$.
            Thus, the chain is present in \(\phi'\) as well and we have
            \(\phi'\wedge x_i \vdash_1 \neg x_j\).
         \item The case of \(i,j\in B\) is analogous to case \(i, j\in A\).
In this case, we use that $i,j \in B$ implies
$\phi \wedge x_i \vdash_1 g$ and $\phi \wedge x_j \vdash_1 g$.
      \end{itemize}
   \end{amocond}
Hence, \(\phi'\) encodes \(\amo_{n-2}\) and
\(\left|\phi\right|\ge|\phi'|+7 \ge \pencQS{n-2}+7\) as required.
\end{proof}

The following proposition will be used in the proof of Theorem~\ref{thm:Main}
to handle the encodings not satisfying the second condition of the restricted
regular form. These are encodings, for which
there is a triangle in the graph whose vertices
are auxiliary variables and the edges are the sets $\PB{i}$ for
$i\in\{1,\ldots,n\}$.

\begin{proposition}
   \label{lem:Triangles}
   Let \(n\ge 4\) and
   let \(\phi(\vek{x},\vek{y})\) be a minimum size 2-CNF
   encoding of \(\amo_n\) with the minimum number of auxiliary variables
and, moreover, assume that $\phi$ is in regular form.
If there are different indices \(r, s, t\in\{1, \ldots, n\}\), such that
   \begin{enumerate}
      \item \(\PB{r}\), \(\PB{s}\), and \(\PB{t}\) are pairwise distinct,
      \item \(\left|\PB{r}\cup\PB{s}\cup\PB{t}\right|= 3\),
   \end{enumerate}
then \(\left|\phi\right|\ge\pencQS{n-2}+6\).
\end{proposition}
\begin{proof}
   \newcommand\auxA{g_\mathrm{A}}
   \newcommand\auxB{g_\mathrm{B}}
   \newcommand\auxC{g_\mathrm{C}}
   Without loss of generality, let us assume that \(r=1\), \(s=2\), and
   \(t=3\) and let \(L=\PA{1}\cup\PA{2}\cup\PA{3}\).
Since $\phi$ is in regular form, we have
\(\left|\PA{1}\right|=\left|\PA{2}\right|=\left|\PA{3}\right|=2\).
Let us distinguish four cases according to the size of \(L\).
Note that there are $|L|-3$ variables in
$\PB{1}\cup\PB{2}\cup\PB{3}$ that occur with both signs in $L$.

   \begin{enumerate}
      \item If \(|L|=3\), then
         \[\PA{1}=\{\auxA,\auxB\},
            \PA{2}=\{\auxA,\auxC\},
            \PA{3}=\{\auxB,\auxC\}\] for some
         \(\auxA,\auxB,\auxC\in\lit{\vek{y}}\). The clause
         \(D=\neg\auxA\vee\neg\auxB\vee\neg\auxC\) is an implicate
         of \(\phi\) because \(\phi\land\neg
            D=\phi\land\auxA\land\auxB\land\auxC\) is unsatisfiable.
         Indeed, any satisfying assignment of
         \(\phi \wedge \auxA \wedge \auxB \wedge \auxC \)
would remain satisfying even if any two of the variables \(x_1,x_2,x_3\)
are changed to 1 which would be in contradiction with the fact that
\(\phi\) encodes \(\amo_n\). On the other hand, since \(\phi\) satisfies~\ref{cond:pc-amo-1},
any two of the literals $\auxA,\auxB,\auxC\in\lit{\vek{y}}$
can be satisfied in a satisfying assignment of $\phi$.
Hence, \(D\) is a prime implicate, which
contradicts \(\phi\) being a 2-CNF formula.

      \item If \(|L|=4\), we can assume due to symmetry that \[\PA{1}=\{\auxA,\auxB\},
            \PA{2}=\{\neg\auxA,\auxC\},
            \PA{3}=\{\auxB,\auxC\}\] for some \(\auxA,\auxB,\auxC\in\lit{\vek{y}}\).
The set $\upclos[\phi]{x_3}$ contains the literals $\auxB$ and $\auxC$
and has a non-empty intersection with each of the sets
$\{\neg \auxA, \neg \auxB\}$, $\{\auxA, \neg \auxC\}$.
One can verify that this implies that for some $e \in \{\auxA, \auxB, \auxC\}$, the
set $\upclos[\phi]{x_3}$ contains both $e$ and $\neg e$. This is
a contradiction with~\ref{cond:pc-amo-1}.

      \item If \(|L|=5\), we can assume due to symmetry that \[\PA{1}=\{\auxA,\auxB\},
            \PA{2}=\{\neg\auxA,\auxC\},
            \PA{3}=\{\neg\auxB,\auxC\}\] for some \(\auxA,\auxB,\auxC\in\lit{\vek{y}}\).
         Consider the formula \(\phi'\) obtained from \(\phi\) by omitting all clauses
         containing variables \(x_1\), \(x_2\), and \(x_3\).
Let \(x_{n+1}\) be a new input variable and consider the formula
         \[\psi=\subst{\phi'}{\auxC}{x_{n+1}}\]
Since \(x_{n+1}\) has no occurence in $\phi'$, this substitution is
the same as renaming $\var{\auxC}$ to $x_{n+1}$ or to $\neg x_{n+1}$,
so that the literal $\auxC$ becomes equal to $x_{n+1}$.
Formula $\psi$ has \(n-2\) input variables \(x_4,\ldots,x_n,x_{n+1}\)
and \(\left|\phi\right|\ge\left|\psi\right|+6\). Let us
check the conditions~\ref{cond:pc-amo-1} and~\ref{cond:pc-amo-2}
         to show that \(\psi\) is an
         encoding of \(\amo_{n-2}\):
         \begin{amocond}
         \item If \(i\in\{4,\ldots,n\}\), a satisfying assignment of \(\phi\wedge x_i\) is
a satisfying assignment of \(\phi'\wedge x_i\). Moreover,
any satisfying assignment of \(\phi\wedge x_2\) satisfies \(\phi\wedge \auxC\)
and, hence, also $\phi'\wedge \auxC$.
Renaming $\auxC$ to $x_{n+1}$ as described above in
$\phi'\wedge x_i$ and $\phi'\wedge \auxC$ yields
the formulas $\psi\wedge x_i$ and $\psi\wedge x_{n+1}$,
so these formulas are satisfiable as well.

         \item Let us check \(\psi\wedge x_i\vdash_1 \neg x_j\)
            for each \(i,j\in\{4,\dots,n+1\}\), \(i<j\). If \(j\le n\), observe that
the series of unit resolutions witnessing $\phi\wedge x_i\vdash_1 \neg x_j$
according to Lemma~\ref{lem:Chains}
does not use the omitted clauses and, hence, is present also in $\phi'$ and $\psi$.
If \(j=n+1\), note that $\upclos[\phi]{x_i}$ contains a negation
of a literal from each of the sets $\PA{1}$, $\PA{2}$, $\PA{3}$
and does not contain complementary literals. This
implies $\neg \auxC \in \upclos[\phi]{x_i}$ and,
hence, $\neg x_{n+1} \in \upclos[\psi]{x_i}$.
         \end{amocond}

      \item If \(|L|=6\), we can assume due to symmetry that \[\PA{1}=\{\auxA,\neg\auxB\},
            \PA{2}=\{\neg\auxA,\auxC\},
            \PA{3}=\{\auxB,\neg\auxC\}\] for some \(\auxA,\auxB,\auxC\in\lit{\vek{y}}\).
Note that the collection of the sets $\{\PA{1}, \PA{2}, \PA{3}\}$ is invariant
under a cyclic shift of the list $(\auxA, \auxB, \auxC)$.
Clearly, $\phi \wedge \auxA \not\vdash_1 \auxB$, since
otherwise $\phi \wedge x_1 \vdash_1 \bot$
in contradiction to~\ref{cond:pc-amo-1}. By symmetry,
all of the following statements are satisfied
\begin{equation} \label{eq:cyclic-clearly-non-derives}
\phi \wedge \auxA \not\vdash_1 \auxB,
\ \phi \wedge \auxB \not\vdash_1 \auxC,
\ \phi \wedge \auxC \not\vdash_1 \auxA\text{.}
\end{equation}
Assume for a contradiction that any two of the statements
$$
\phi\wedge \auxA \vdash_1 \auxC,
\ \phi\wedge \auxC \vdash_1 \auxB,
\ \phi\wedge \auxB \vdash_1 \auxA
$$
are satisfied. For each pair of these statements, we get a contradiction
with~\eqref{eq:cyclic-clearly-non-derives}.
Hence, at most one of these statements is satisfied. It follows that
we can assume without loss of generality
\begin{equation} \label{eq:cyclic-two-non-derives}
\phi\wedge \auxA \not\vdash_1 \auxC,
\ \phi\wedge \auxB \not\vdash_1 \auxA
\end{equation}
The following claim will be used later.

\smallskip
\noindent\textbf{Claim.}
If $k \ge 4$, then the set $\upclos[\phi]{x_k}$ contains either
all of the literals $\auxA, \auxB, \auxC$ or all of the literals
$\neg \auxA, \neg \auxB, \neg \auxC$.
\begin{proof}
Since $\phi \wedge x_k \vdash_1 \neg x_1$, we have
$\auxB \in \upclos[\phi]{x_k}$ or $\neg\auxA \in \upclos[\phi]{x_k}$.
In the first case, the set $\upclos[\phi]{x_k}$ contains
$\auxC$ to derive $\neg x_3$ and contains $\auxA$ to derive $\neg x_2$.
In the second case, the set $\upclos[\phi]{x_k}$ contains
$\neg\auxC$ to derive $\neg x_2$ and contains $\neg\auxB$ to derive $\neg x_3$.
\end{proof}

By Lemma~\ref{lem:aux-var-occurrences} we get that the variable
\(\var{\auxA}\) occurs in at least 5 clauses.
Thus, the formula \(\psi\)
         obtained from \(\phi\) by omitting all clauses containing
         some of the literals \(\neg x_1\), \(\neg x_2\), \(\auxA\),
         and \(\neg \auxA\) satisfies
         \(\left|\phi\right|\ge\left|\psi\right|+7\). It remains to
         check the conditions~\ref{cond:pc-amo-1} and~\ref{cond:pc-amo-2}
         to show that
         \(\psi(x_3,\ldots,x_n)\) is an encoding of \(\amo_{n-2}\):
         \begin{amocond}
         \item As \(\psi\) is a subset of \(\phi\), each satisfying
assignment of \(\phi \wedge x_i\) for $i\in\{3,\ldots,n\}$
can be restricted to a satisfying assignment of \(\psi \wedge x_i\).
         \item Let us check \(\psi\wedge x_i\vdash_1 \neg x_j\)
for each \(i,j\in\{3,\dots,n\}\), $i \not= j$.
By Lemma~\ref{lem:Chains-2cnf-sym}, we can assume $i < j$.
Let us consider separately the cases $i=3$ and $i \ge 4$.
\noindent\textbf{Case} $i \ge 4$.\\
Fix a series of unit
resolutions witnessing $\phi \wedge x_i\vdash_1 \neg x_j$
according to Lemma~\ref{lem:Chains}.
This series does not use any clause containing the literals $\neg x_1$, $\neg x_2$.
If this series does not derive a literal
on the variable $\var{\auxA}$, then $\psi \wedge x_i\vdash_1 \neg x_j$
and we are done.
Assume, the series of resolutions derives a literal \(h\in\{\auxA,\neg \auxA\}\).
In order to prove $\psi \wedge x_i\vdash_1 \neg x_j$, we split this
series into the two parts at the occurrence of the literal $h$ and replace
each of these parts by another series of resolutions in such a way that
they can be combined via a literal on the variable $\var{\auxC}$ and do
not derive a literal on the variable $\var{\auxA}$. Using this,
we obtain a series witnessing $\psi \wedge x_i\vdash_1 \neg x_j$.

By Lemma~\ref{lem:Chains-2cnf-sym}, we have \(\phi\wedge x_i \vdash_1 h\)
and \(\phi\wedge x_j \vdash_1 \neg h\).
Let us prove that for each $k \in \{4,\ldots,n\}$, we have
\begin{itemize}
\item \(\phi\wedge x_k \vdash_1 \auxA\) implies \(\psi\wedge x_k \vdash_1 \auxC\)
\item \(\phi\wedge x_k \vdash_1 \neg\auxA\) implies \(\psi\wedge x_k \vdash_1 \neg\auxC\)
\end{itemize}

If \(\phi\wedge x_k \vdash_1 \auxA\), the set $\upclos[\phi]{x_k}$
contains $\auxC$ by the claim above. Consider a series of resolutions
witnessing $\phi\wedge x_k \vdash_1 \auxC$.
As $\auxA \in \upclos[\phi]{x_k}$, this series does not derive \(\neg\auxA\).
By~\eqref{eq:cyclic-two-non-derives}, this series
does not derive \(\auxA\). Together, we have
$\psi\wedge x_k \vdash_1 \auxC$.

If \(\phi\wedge x_k \vdash_1 \neg\auxA\), the set $\upclos[\phi]{x_k}$
contains $\neg \auxC$ by the claim above.
Consider a series of resolutions witnessing $\phi \wedge x_k \vdash_1 \neg \auxC$.
As $\neg\auxA \in \upclos[\phi]{x_k}$, this series does not derive \(\auxA\).
By~\eqref{eq:cyclic-clearly-non-derives} and Lemma~\ref{lem:Chains-2cnf-sym},
this series does not derive \(\neg\auxA\). Together, we have
$\psi\wedge x_k \vdash_1 \neg\auxC$.

If \(h=\auxA\), then \(\phi\wedge x_i\vdash_1\auxA\) and
\(\phi\wedge x_j\vdash_1\neg \auxA\). By the above implications,
we have \(\psi\wedge x_i\vdash_1\auxC\) and \(\psi\wedge x_j\vdash_1\neg \auxC\).
Using this and Lemma~\ref{lem:Chains-2cnf-sym}, we obtain \(\psi\wedge x_i\vdash_1\neg x_j\).
If \(h=\neg\auxA\), we obtain $\psi\wedge x_i\vdash_1\neg x_j$ by
a similar argument using a series of resolutions deriving $\neg \auxC$
as an intermediate step.

\noindent\textbf{Case} $i = 3$.\\
We have either $\phi \wedge x_j \vdash_1 \neg\auxB$ or
$\phi \wedge x_j \vdash_1 \auxC$.
Assume $\phi \wedge x_j \vdash_1 \neg\auxB$
and consider a series of unit resolutions witnessing this.
By the claim above,
this series does not derive the literal $\auxA$.
By~\eqref{eq:cyclic-two-non-derives}
and Lemma~\ref{lem:Chains-2cnf-sym}, it does not
derive $\neg\auxA$. Hence, $\psi \wedge x_j \vdash_1 \neg\auxB$
and we have $\psi \wedge x_3 \vdash_1 \neg x_j$.

Assume $\phi \wedge x_j \vdash_1 \auxC$ and consider a series
of unit resolutions witnessing this. By the claim above,
this series does not derive the literal $\neg\auxA$.
By~\eqref{eq:cyclic-two-non-derives}, it does not
derive $\auxA$. Hence, $\psi \wedge x_j \vdash_1 \auxC$
and we have $\psi \wedge x_3 \vdash_1 \neg x_j$.
\qedhere
         \end{amocond}
   \end{enumerate}

\end{proof}

It remains to estimate the size of an encoding of $\amo_n$ that
satisfies both the additional conditions of the restricted
regular form.

\begin{proposition} \label{lem:quad-regular}
   Let \(n\geq 4\) and let \(\phi(\vek{x}, \vek{y})\) be a minimum
   size 2-CNF encoding of \(\amo_n\) with the minimum number of
   auxiliary variables
and, moreover, assume that $\phi$ is in regular form.
If $\phi$ satisfies the following two conditions
   \begin{itemize}
      \item The sets \(\PB{1}, \ldots, \PB{n}\) are pairwise distinct.
      \item For any three different indices
   \(r, s, t\in\{1, \ldots, n\}\) we have
   \(\left|\PB{r}\cup\PB{s}\cup\PB{t}\right|> 3\),
   \end{itemize}
   then \(|\phi|\geq 2n+2\sqrt{n}-3\).
\end{proposition}
\begin{proof}
   Consider the undirected graph \(G=(\vek{y},E)\), whose vertices are
   auxiliary variables and different variables $u,v \in \vek{y}$ are
   connected by an edge $(u,v) \in E$ if and only if
   \((g\vee h)\in\phi\) for some \(g\in\lit{u}\) and \(h\in \lit{v}\).
   Let \(\mathcal{K}\) denote the set of the connected components of \(G\).
   Note that the elements of \(\mathcal{K}\) are sets of variables,
   which form a partition of \(\vek{y}\). For \(i \in \{1,\dots, n\}\), let
   \[
      \PC{i}=\{ K\in\mathcal{K}\;|\; \PB{i}\cap
         K\neq\emptyset\}\text{.}
   \]
   Since $\phi$ is in regular form, $|\PC{i}| \le 2$.
   Fix $i \not= j$ and let $e_0, e_1, \dots, e_p\in \lit{\vek{x}\cup\vek{y}}$
   be a sequence of literals derived in a series of resolutions
   witnessing $\phi \wedge x_i \vdash_1 \neg x_j$ according
   to Lemma~\ref{lem:Chains}. Hence,
   $e_1 \in \PA{i}$ and $\neg e_{p-1} \in \PA{j}$. Then,
   the variables $\var{e_1}, \ldots, \var{e_{p-1}}$ form a path in \(G\)
   between a vertex in one of the components in $\PC{i}$ and
   a vertex in one of the components in $\PC{j}$. Since all the
   vertices of a path belong to the same connected component,
   we have \(\PC{i}\cap \PC{j}\neq\emptyset\).

   Let us prove that \(|\mathcal{K}|\le 3\) by contradiction.
   Assume \(|\mathcal{K}|\ge 4\) and distinguish two cases:
   \begin{enumerate}
      \item There exists \(K_1 \in \mathcal{K}\), such that
         \(K_1\in \PC{i}\) for all \(i \in \{1,\dots, n\}\).
         Choose different components \(K_2,K_3\in\mathcal{K}\) different from \(K_1\),
         choose variables \(u\in K_2, v\in K_3\),
         and consider the formula \(\phi'=\subst{\phi}{u}{v}\).
         We show that \(\phi'\) satisfies the
         conditions~\ref{cond:pc-amo-1} and~\ref{cond:pc-amo-2},
         thus it encodes \(\amo_n\) in contradiction to the assumption
         that \(\phi\) has the minimum number of auxiliary variables.
         \begin{amocond}
         \item For every \(i \in \{1,\dots,n\}\), we prove that there
            is a satisfying assignment \(\alpha\) of \(\phi\wedge x_i\),
            such that \(\alpha(u)=\alpha(v)\). This implies that
            \(\phi'\wedge x_i\) is satisfiable as well.

            Clearly, \(K_2\not\in \PC{i}\) or \(K_3\not\in \PC{i}\),
            because \(|\PC{i}| \le 2\) and \(K_1\in \PC{i}\). Due to the symmetry
            between $K_2$ and $K_3$ and the variables $u$ and $v$,
            we can assume \(K_2\not\in \PC{i}\). As \(\phi\wedge x_i\) is satisfiable,
            there is a literal \(e\in \lit{v}\) such that \(\phi\wedge x_i \wedge e\)
            is satisfiable.
            Lemma~\ref{lem:PositiveInpVars} and the assumption that $\phi$
            is in regular form imply that the set $\upclos[\phi]{x_i, e}$
            contains only the literal $x_i$, literals $\neg x_j$ for $j \not=i$,
            and literals on some of the auxiliary variables
            in the components contained in $\PC{i}$ and in $K_3$. In particular, we have
            \(\upclos[\phi]{x_i, e}\cap \lit{K_2} = \emptyset \).
            By Lemma~\ref{lem:minimum-2-CNF-AMO}, $\phi$ is propagation complete.
            Hence, none of the literals $u$ and $\neg u$ is an implicate
            of \(\phi\wedge x_i \wedge e\).
            Consequently, there exists a satisfying assignment \(\alpha\) of the formula
            \(\phi\wedge x_i\wedge e\) such that \(\alpha(u)=\alpha(v)\) as required.
         \item Since $\phi'\wedge x_i$ is satisfiable for every \(i\in\{1,
                  \ldots, n\}\), Lemma~\ref{lem:identification-propagation}
            for the formula $\phi \wedge x_i$ and the substitution $[u \leftarrow v]$
            implies $\phi' \wedge x_i \vdash_1 \neg x_j$ for all $j \not=i$.
         \end{amocond}

      \item Assume that for each \(K \in \mathcal{K}\)
         there exists \(i \in \{1,\dots,n\}\) such that \(K\not\in \PC{i}\).
         It follows that \(\left|\PC{i}\right|=2\) for each \(i \in \{1,\dots, n\}\).
         Indeed, if \(\PC{j}=\{K\}\) for some \(j\) and \(K\in\mathcal{K}\),
         then $K \in \PC{i}$ for all $i$, since $\PC{i}\cap \PC{j}\neq\emptyset$.

         By the assumptions, there are $i,j$, such that $\PC{i} \not= \PC{j}$.
         Since $\PC{i}$ and $\PC{j}$ have a non-empty intersection, there
         are $K_1,K_2,K_3 \in \mathcal{K}$, such that
         \begin{align*}
            \PC{i} & = \{K_1, K_2\}\text{,}\\
            \PC{j} & = \{K_1, K_3\}\text{.}
         \end{align*}
         By the assumptions, there is $k$, such that $K_1 \not\in \PC{k}$.
         Since $\PC{k}$ has a non-empty intersection with both $\PC{i}$
         and $\PC{j}$, we have
         \begin{align*}
            \PC{k} & = \{K_2, K_3\}\text{.}
         \end{align*}
         Since $|\mathcal{K}|\ge 4$, there is
         $K_4 \in \mathcal{K} \setminus \{K_1, K_2, K_3\}$
         and, moreover, there is $l$, such that $\PC{l} = \{K_4, K_5\}$
         for some $K_5 \in \mathcal{K}$.
         Since $\PC{l}$ has a non-empty intersection with
         each of the sets $\PC{i}$, $\PC{j}$, $\PC{k}$, each of these
         sets contains $K_5$. This is a contradiction, since these
         sets have no common element.
   \end{enumerate}

   Let \(\psi \subset \phi\) denote the set of clauses of \(\phi\) that contain
   only variables from \(\vek{y}\). As \(G\) has at most three connected components
   and the number of edges of $G$ is $|\psi|$,
   we have \(\left|\psi\right|\ge \left|\vek{y}\right|-3\).

   Consider an undirected graph $G'$ with vertices \(\vek{y}\), whose edges are
   the sets $\PB{1},\dots,\PB{n}$. By assumption, \(G'\) contains $n$ edges
and does not contain a triangle.
Since $G'$ contains $|\vek{y}|$ vertices,
Mantel's theorem (a special case of Tur\' an's theorem) implies
\(n \le \frac{1}{4}|\vek{y}|^2\).
   Thus, \(\left|\vek{y}\right|\ge 2\sqrt{n}\) and  \(\left|\psi\right|  \ge 2\sqrt{n}-3\).
   Finally, we obtain $|\phi|=2n+|\psi| \ge 2n + 2\sqrt{n}-3$ as
   required.
\end{proof}

The following theorem presents the second main result of this paper.
Namely, it shows part~\ref{enum:lower-bound:1} for \(\pencQS{n}\)
and part~\ref{enum:lower-bound:3} of Theorem~\ref{thm:lower-bound}.

\begin{theorem} \label{thm:Main}
For \(n\ge 3\), the minimum size $\pencQS{n}$ of a 2-CNF encoding of $\amo_n$ satisfies
\begin{enumerate}
   \item If $n \le 10$, then $\pencQS{n} = 3n-6$.
   \item
      If $n \ge 9$, then
      $\pencQS{n} \ge 2n + 2\sqrt{n} - 3$.
\end{enumerate}
\end{theorem}

\begin{proof}
If $3 \le n \le 8$, the conclusion follows from
Lemma~\ref{lem:minimum-2-CNF-AMO}, Theorem~\ref{thm:general-lb-3}
and Lemma~\ref{lem:simple-upper-bound}.
For the rest of the proof, let $\phi$ be a minimum size 2-CNF
encoding of $\amo_n$ that, morever, has the minimum number
of auxiliary variables among such encodings.
In particular, we have $|\phi|=\pencQS{n}$.
We first analyze the cases \(n=9\) and \(n=10\).
The upper bound $\pencQS{n} \le 3n-6$ follows from Lemma~\ref{lem:simple-upper-bound}.
The lower bound $\pencQS{n} \ge 3n-6$ for $n=9$ and $n=10$ can be proven as follows.

\begin{itemize}
   \item Assume \(n=9\). If \(\phi\) is not in regular
      form, we have by Theorem~\ref{thm:structure-general-CNF-2}
      \begin{equation*}
         |\phi| \geq 3+\pencQS{n-1} = 3+\pencQS{8} = 21\text{.}
      \end{equation*}
      If \(\phi\) is in regular form and either assumptions of
      Proposition~\ref{lem:Equal-PB-Cont} or assumptions of
      Proposition~\ref{lem:Triangles} are satisfied, then we have
      \begin{equation*}
         |\phi| \geq 6+\pencQS{n-2} = 6+\pencQS{7} = 21\text{.}
      \end{equation*}
It remains to consider \(\phi\) in regular form satisfying the assumptions
of Proposition~\ref{lem:quad-regular}, for which we have
      \begin{equation*}
         |\phi| \geq 2n+2\sqrt{n}-3=21\text{.}
      \end{equation*}
Altogether, $\pencQS{n}=|\phi| \ge 21 = 3n-6$.
   \item Assume \(n=10\). If \(\phi\) is not in regular
      form, we have by Theorem~\ref{thm:structure-general-CNF-2}
      \begin{equation*}
         |\phi| \ge 3 + \pencQS{n-1} = 3+\pencQS{9}=24\text{.}
      \end{equation*}
      If \(\phi\) is in regular form and either assumptions of
      Proposition~\ref{lem:Equal-PB-Cont} or assumptions of
      Proposition~\ref{lem:Triangles} are satisfied, then we have
      \begin{equation*}
         |\phi| \geq 6+\pencQS{n-2} = 6+\pencQS{8} = 24\text{.}
      \end{equation*}
      It remains to consider \(\phi\) in regular form satisfying
the assumptions of Proposition~\ref{lem:quad-regular}.
Since $|\phi|$ is an integer, we have
      \begin{equation*}
         |\phi| \geq \left\lceil 2n + 2 \sqrt{n}-3 \right\rceil =
 \left\lceil 17 + 2\cdot \sqrt{10} \right\rceil=24\text{.}
      \end{equation*}
Altogether, $\pencQS{n}=|\phi| \ge 24 = 3n-6$.
\end{itemize}

We prove $\pencQS{n} \ge 2n + 2\sqrt{n} - 3$ for $n \ge 9$ by
induction. Since $3n-6 \ge 2n + 2\sqrt{n} - 3$ for $n=9$ and $n=10$,
the lower bound $\pencQS{n} \ge 2n + 2\sqrt{n} - 3$ is already
proven for these values of $n$.
Assume \(n>10\) for the rest of the proof and consider
three cases concerning the structure of $\phi$.

If $\phi$ is not in regular form, then
Theorem~\ref{thm:structure-general-CNF-2} and the
induction hypothesis imply
\begin{align*}
   \left|\phi\right| & \ge 3+\pencQS{n-1}\\
   & \ge  3+2(n-1)+2\sqrt{n-1}-3\\
   & =  2n+1+2\sqrt{n-1}-3\\
   & \ge  2n+2\sqrt{n}-3
\end{align*}
because \(1+2\sqrt{n-1}\ge 2\sqrt{n}\) holds whenever
\(n\ge 2\).

If \(\phi\) is in regular form and either assumptions of
Proposition~\ref{lem:Equal-PB-Cont} or assumptions of
Proposition~\ref{lem:Triangles} are satisfied, then using the induction
hypothesis, we have
\begin{align*}
   \begin{split}
      \left|\phi\right| & \ge 6+\pencQS{n-2}\\
      & \ge  6+2(n-2)+2\sqrt{n-2}-3\\
      & =  2n+2+2\sqrt{n-2}-3\\
      & \ge  2n+2\sqrt{n}-3
   \end{split}
\end{align*}
because \(2+2\sqrt{n-2}\ge 2\sqrt{n}\) holds whenever $n \ge 3$.

It remains to consider \(\phi\) in regular form satisfying the assumptions of
Proposition~\ref{lem:quad-regular}. This lemma implies
directly the required lower bound.
\end{proof}

   \section{Further Research}
\label{sec:further}

Since 2-CNF formulas are closed under resolution,
a function can have a 2-CNF encoding only if it is 
expressible by a 2-CNF formula. The function \(\amo_n\) can be
represented by an anti-monotone 2-CNF formula. It is quite natural to ask if
there is a minimum size PC encoding of \(\amo_n\) that is a 2-CNF
formula, or at least a CNF formula
without positive occurrences of the input variables. More generally,
we can pose the following questions.
Note that by Lemma~\ref{lem:PositiveInpVars},
a positive answer to Question~\ref{question:1} implies a negative
answer to Question~\ref{question:2}.

\begin{question}\label{question:1}
Assume, $f(\vek{x})$ is a boolean function expressible by
a monotone or anti-monotone 2-CNF formula.
Is there a PC encoding of the function $f$ of minimum size, which is,
moreover, a 2-CNF formula?
\end{question}

\begin{question}\label{question:2}
   Is there a PC encoding of $\amo_n$ of minimum size,
   which contains a positive occurence of an input variable?
\end{question}

We expect a negative answer to Question~\ref{question:2}. However,
for every sufficiently large $n$, there is an irredundant prime PC encoding
of $\amo_n$ of size $\Theta(n^2)$ that contains positive occurences
of the input variables. Let us briefly present an example of such a formula.
Let $A$, $B$, $C$, $D$ be non-empty sets that form a partition
of $\{1,\ldots,n\}$ and consider auxiliary variables $y_1,\ldots,y_5$.
Let $E$ be the set of the edges of the complete graph on the vertices $\{1,\ldots,n\}$
except of the edges contained in the bipartite graph between $A$ and $B$
and the bipartite graph between $C$ and $D$. The formula
\[
\begin{array}{l}
\bigwedge_{\{i,j\} \in E} (\neg x_i \vee \neg x_j) \ \wedge\\
(y_1 \vee y_2 \vee y_3 \vee y_4 \vee y_5) \ \wedge\\
\left(\neg y_1 \vee \bigvee_{i \in A} x_i\right) \wedge \bigwedge_{i \in B} (\neg y_1 \vee \neg x_i) \ \wedge\\
\left(\neg y_2 \vee \bigvee_{i \in B} x_i\right) \wedge \bigwedge_{i \in A} (\neg y_2 \vee \neg x_i) \ \wedge\\
\left(\neg y_3 \vee \bigvee_{i \in C} x_i\right) \wedge \bigwedge_{i \in D} (\neg y_3 \vee \neg x_i) \ \wedge\\
\left(\neg y_4 \vee \bigvee_{i \in D} x_i\right) \wedge \bigwedge_{i \in C} (\neg y_4 \vee \neg x_i) \ \wedge\\
\bigwedge_{i=1}^n (\neg y_5 \vee \neg x_i)
\end{array}
\]
is a prime PC encoding of $\amo_n$ containing a positive occurence of
each input variable. Moreover, removing any of the clauses
with a positive occurence of an input variable leads to a formula
that is not a PC encoding of $\amo_n$. For example, if we remove the
clause $\neg y_1 \vee \bigvee_{i \in A} x_i$, then one can
obtain a satisfying assignment inconsistent with the function
$\amo_n$ as follows. Choose $i \in C$,
$j \in D$, set $x_i=x_j=1$,
$x_k=0$ for all $k\in \{1,\ldots,n\} \setminus \{i,j\}$,
$y_1=1$, and $y_2 = \cdots = y_5 = 0$.

It is plausible to assume that for every $n \ge 3$, there is a minimum size PC encoding
of $\exone_n$, which has the form
$$
(x_1\lor\cdots\lor x_n) \wedge \phi(\vek{x}, \vek{y})\text{,}
$$
where $\phi(\vek{x}, \vek{y})$ is a minimum size PC encoding of $\amo_n$.
This suggests the following conjecture, which is a strengthening
of Proposition~\ref{prop:relations-for-sizes}.

\begin{conjecture}\label{conjecture:1}
   Let \(\phi\) be a propagation complete encoding of \(\amo_n\)
of minimum size
   and let \(\psi\) be a propagation complete encoding of \(\exone_n\)
of minimum size
   for \(n\geq 2\). Then \(|\psi|=|\phi|+1\).
\end{conjecture}

The requirement that an encoding is propagation complete can
be relaxed to unit refutation completeness introduced in~\cite{V94}. An encoding
$\phi(\vek{x}, \vek{y})$ is unit refutation complete, if for every
partial assignment of $\vek{x}$ that makes $\phi$ unsatisfiable,
a contradiction can be derived by unit propagation.
This makes no difference for a 2-CNF formula. A minimum size
2-CNF encoding is propagation complete and, hence,
unit refutation complete. It follows that for 2-CNF,
the minimum size of a unit refutation complete encoding and
the minimum size of a propagation complete encoding of
the same function are the same. However,
for general CNF formulas, a unit refutation complete encoding
can be smaller than the smallest propagation complete encoding
and, hence, a lower bound on the size of
a unit refutation complete encodings is harder to prove.

\begin{question}
Is the minimum size of a unit refutation complete
encoding of $\amo_n$ or \(\exone_n\) at least $2n + \Omega(\sqrt{n})$?
\end{question}

\section{Conclusion}
\label{sec:conclusion}

We have shown that any propagation complete encoding of the
\(\amo_n\) or \(\exone_n\) constraint for \(n\geq 9\) contains at least
\(2n+\sqrt{n}-2\) clauses. This shows that the best known upper
bound of \(2n + 4 \sqrt{n} + O(\sqrt[4]{n})\) clauses achieved by product
encoding introduced by Chen~\cite{amoCHE_2010} is essentially best
possible. Let us point out that the product encoding is an encoding
of $\amo_n$ in regular form which is the notion playing central role
in our proof.

For the special case of 2-CNF encodings, we have shown
for \(n\geq 9\) a better lower bound \(2n+2\sqrt{n}-3\). This case is
important, because the encodings that appear in the literature are
2-CNF formulas including the product encoding mentioned above.

We have also shown that for \(3 \le n \le 8\), the number
of clauses in a propagation complete encoding of $\amo_n$ is at least \(3n-6\).
This number of clauses is achieved by sequential encoding and therefore
in this case the lower and upper bound match for both general CNF
formulas and 2-CNF formulas.

\section*{Acknowledgments}
Petr Ku{\v c}era was supported by the Czech Science Foundation (grant
GA15-15511S).
Petr~Savick\'y was supported by CE-ITI and GA\v CR under
the grant GBP202/12/G061 and by the institutional
research plan RVO:67985807.

   \bibliography{ms}
   \bibliographystyle{plain}
\end{document}